\DeclareSymbolFont{AMSb}{U}{msb}{m}{n}
\documentclass[reqno,centertags,10pt,a4paper,noamsfonts]{amsart}
\pdfoutput=1
\usepackage[svgnames]{xcolor}
\usepackage{graphicx}
\usepackage[english]{babel}
\usepackage[babel=true,expansion=alltext,protrusion=alltext-nott,final]{microtype}
\usepackage[margin={3cm},marginparwidth={2.5cm}]{geometry}
\usepackage[utf8]{inputenc}
\usepackage{textcomp}
\usepackage[sb]{libertine}
\usepackage[libertine,bigdelims,vvarbb]{newtxmath}
\useosf
\usepackage[varqu,varl]{inconsolata}
\usepackage[supsfam=LibertinusSerif-Sup,%
       supscaled=1.2,%
       raised=-.13em]{superiors}
\usepackage{mathtools}
\usepackage[T1]{fontenc}
\usepackage{textcomp}
\usepackage{amsthm}

\usepackage[cal=cm]{mathalpha}
\usepackage[inline]{enumitem}
\numberwithin{equation}{section}
\usepackage[normalem]{ulem}

\usepackage{float}  
\usepackage{caption}
\usepackage{subcaption}
\usepackage{xspace}         
\usepackage[scientific-notation=true]{siunitx}      
\usepackage{pstricks}
\usepackage{tikz}
\usetikzlibrary{positioning} 
\usetikzlibrary{calc}
\usepackage{pgfplots}
\pgfplotsset{width=10cm,compat=1.9}
\usepackage{bm}
\usepackage{sidecap}
\usepackage{graphicx,color}

\usepackage{amsmath}
\DeclareFontFamily{U}{mathx}{}
\DeclareFontShape{U}{mathx}{m}{n}{<-> mathx10}{}
\DeclareSymbolFont{mathx}{U}{mathx}{m}{n}
\DeclareMathAccent{\widehat}{0}{mathx}{"70}
\DeclareMathAccent{\widecheck}{0}{mathx}{"71}

\providecommand{\mr}[1]{\href{http://www.ams.org/mathscinet-getitem?mr=#1}{MR~#1}}
\providecommand{\zbl}[1]{\href{https://zbmath.org/?q=an:#1}{Zbl~#1}}

\providecommand{\doi}[2]{\href{https://doi.org/#1}{#2}}
\definecolor{light_gray}{gray}{0.75}
\definecolor{lighter_gray}{gray}{0.5}
\colorlet{light_blue}{blue!20}
\definecolor{dark_green}{rgb}{0.0, 0.6, 0.0}
\definecolor{royal_blue}{rgb}{0.0, 0.22, 0.66}
\definecolor{salmon}{rgb}{1.0, 0.55, 0.41}
\definecolor{gold}{rgb}{0.8, 0.63, 0.21}
\definecolor{navy_blue}{rgb}{0.0, 0.0, 0.5}
\definecolor{crimson}{rgb}{0.79, 0.0, 0.09}
\definecolor{amethyst}{rgb}{0.6, 0.4, 0.8}
\definecolor{alizarin}{rgb}{0.82, 0.1, 0.26}
\definecolor{amaranth}{rgb}{0.9, 0.17, 0.31}
\definecolor{azure}{rgb}{0.0, 0.5, 1.0}
\definecolor{canaryyellow}{rgb}{0.82, 0.41, 0.12}
\definecolor{carrotorange}{rgb}{0.8, 0.33, 0.0}
\definecolor{cadmiumgreen}{rgb}{0.0, 0.42, 0.24}
\definecolor{copper}{rgb}{0.72, 0.45, 0.2}
\definecolor{aqua}{rgb}{0.5, 1.0, 0.83}
\definecolor{awesome}{rgb}{1.0, 0.13, 0.32}
\definecolor{candyapplered}{rgb}{1.0, 0.03, 0.0}
\definecolor{caribbeangreen}{rgb}{0.0, 0.8, 0.6}
\definecolor{indigo}{rgb}{0.0, 0.25, 0.42}

\DeclareMathOperator{\weaklystar}{\rightharpoonup\kern-2.2ex ^* \, \,}

\def\XXint#1#2#3{{\setbox0=\hbox{$#1{#2#3}{\int}$ }
\vcenter{\hbox{$#2#3$ }}\kern-.6\wd0}}

\newcommand{\R}{\mathbb R}
\newcommand{\N}{\mathbb N}
\newcommand{\Z}{\mathbb Z}

\newcommand\norm[1]{\lVert #1 \rVert}
\newcommand\bignorm[1]{\big\lVert #1 \big\rVert}

\newcommand\inner[1]{\langle #1 \rangle}

\newcommand{\mL}{\mathrm{L}}

\renewcommand{\phi}{\varphi}

\newcommand{\mH}{\mathrm{H}}
\newcommand{\mW}{\mathrm{W}}

\newcommand{\T}{\mathbb{T}}

\theoremstyle{plain}
\newtheorem{theorem}{Theorem}[section]
\newtheorem{proposition}[theorem]{Proposition}

\newtheorem{lemma}[theorem]{Lemma}
\newtheorem{assumption}[theorem]{Assumption}
\newtheorem*{theorem*}{Theorem}

\theoremstyle{definition}
\newtheorem{definition}[theorem]{Definition}
\newtheorem{remark}[theorem]{Remark}
\newtheorem*{remark*}{Remark}

\usepackage[pagebackref=false]{hyperref}
\hypersetup{
	linktoc=all,
	bookmarksdepth=3,
	colorlinks=true,
	linkcolor=Green,citecolor=Orange,urlcolor=DarkBlue,
	pdftitle={Grand canonical DFT in 1D},
	pdfauthor={Thiago Carvalho Corso}, 
	pdfsubject={},
	pdfkeywords={}} 
\begin{document}
\numberwithin{table}{section}
%%%%%%%%%%%%%%%%%%%%%%%%%%%%%%%%%%%%%%%%%%%%%%%%%%
\title[Strictly correlated electrons in a quantum ring]{Strictly correlated electrons in a quantum ring: from Kohn-Sham to Kantorovich potentials}

%%%%%%%%%%%%%%%%%%%%%%%%%%%%%%%%%%%%%%%%%%%%%%%%%%
\author[T.~Carvalho~Corso]{Thiago Carvalho Corso}
\address[T.~Carvalho Corso]{Institute of Applied Analysis and Numerical Simulation, University of Stuttgart, Pfaffenwaldring 57, 70569 Stuttgart, Germany}
\email{thiago.carvalho-corso@mathematik.uni-stuttgart.de}

%%%%%%%%%%%%%%%%%%%%%%%%%%%%%%%%%%%%%%%%%%%%%%%%%%
%% 49Q20 - Variational problems in a geometric measure-theoretic setting
%% 	35Q40  - PDEs in connection with quantum mechanics
%% 	81Q05  - Closed and approximate solutions to the Schrödinger, Dirac, Klein-Gordon and other equations of quantum mechanics
%% 81Q35  - Quantum mechanics on special spaces: manifolds, fractals, graphs, lattices 

\date{\today}
\thanks{\emph{Funding information}:  DFG -- Project-ID 442047500 -- SFB 1481.  \\[1ex]
\textcopyright 2026 by the author. Faithful reproduction of this article, in its entirety, by any means is permitted for noncommercial purposes.}
%%%%%%%%%%%%%%%%%%%%%%%%%%%%%%%%%%%%%%%%%%%%%%%%%%
%\dedicatory{}
%%%%%%%%%%%%%%%%%%%%%%%%%%%%%%%%%%%%%%%%%%%%%%%%%%
\begin{abstract} Our goal in this paper is twofold. First, we characterize the class of pairwise interactions for which the Seidl conjecture on the structure of optimal plans for the symmetric multimarginal optimal transport problem with one-dimensional marginal holds. This extends previous results by Colombo, De Pascale, and Di Marino \cite{CDD15}, which treated the case of translation-invariant, convex and decreasing interactions. In particular, our results apply to physically relevant interactions for electrons living on a quantum ring. The second main goal of the paper is to rigorously derive the leading order asymptotics of the adiabatic connection potential for strongly interacting systems. More precisely, we show that for electrons in a quantum ring (or one-dimensional interval), not only the Lieb density functional converges to the optimal transport (or strictly correlated) functional in the semiclassical limit, but also the representing potential converges to a regular Kantorovich potential. As an intermediate step, we also extend previous results on the strongly interacting limit of the Lieb functional to periodic systems in arbitrary dimensions.
\end{abstract}
%%%%%%%%%%%%%%%%%%%%%%%%%%%%%%%%%%%%%%%%%%%%%%%%%%
\setcounter{tocdepth}{1}
\maketitle
%\tableofcontents
%%%%%%%%%%%%%%%%%%%%%%%%%%%%%%%%%%%%%%%%%%%%%%%%%%
\setcounter{secnumdepth}{2}
%%%%%%%%%%%%%%%%%%%%%%%%%%%%%%%%%%%%%%%%%%%%%%%%%%
\section{Introduction}

\subsection{Motivation}

In \cite{CDD15}, Colombo, De Pascale, and Di Marino have shown that, for pairwise interaction potentials $w$ that are convex and decreasing, the multimarginal optimal transport problem 
\begin{align}
    F_{\rm OT}(\rho) \coloneqq  \min \left\{\int_{\R^n} c_n(x) \mathrm{d} \gamma(x)  : \gamma \in \Pi(\rho)\right\} \label{eq:MMOT}
\end{align}
with cost function
\begin{align}
    c_n(x_1,...,x_n) = \sum_{i \neq j} c_2(x_i,x_j) = \sum_{i \neq j} w(x_i-x_j), \label{eq:cost function}
\end{align}
where $\Pi(\rho)$ denotes the set of probability measures in $\R^n$ with each marginal equal to $\rho \in \mathcal{P}(\R)$, has an explicit minimizer in terms of $\rho$. More precisely, for any non-atomic $\rho \in \mathcal{P}(\R)$, we let $-\infty = d_0 < d_1 < ... < d_n = +\infty$ be such that
\begin{align*}
    \rho([d_i,d_{i+1}] ) = 1/n, \quad i = 0,1,..., n-1,
\end{align*}
and let $T:\R \rightarrow \R$ be the unique (up to $\rho$ null sets) function that is increasing on each interval $(d_i,d_{i+1})$ and satisfies
\begin{align*}
    &T^\#(\mathbb{1}_{[d_i,d_{i+1}]} \rho ) = \mathbb{1}_{[d_{i+1},d_{i+2}]} \rho, \quad i=0,1,..., n-2\\
    &T^\#(\mathbb{1}_{[d_{n-1},d_{n}]} \rho) = \mathbb{1}_{[d_0,d_1]} \rho.
\end{align*}
Then the push-forward measure
\begin{align}
    \gamma_\rho = (\mathrm{id}, T, T^{(2)},...,T^{(n-1)})^\# \rho \in \Pi(\rho), \quad \mbox{where $T^{(k)} = \overbrace{T\circ T ... \circ T}^{k \text{~times}}$,}\label{eq:optimizer}
\end{align}
 is an optimizer of~\eqref{eq:MMOT}. Moreover, if $w$ is strictly convex, then the symmetrization of~\eqref{eq:optimizer} is the unique minimizer of~\eqref{eq:MMOT} among the symmetric measures.

This result was conjectured by Seidl in \cite{Sei99} and plays an important role in the strongly interacting limit of density functional theory (DFT) \cite{SGS07,GSV09}. More precisely, in the limit $\varepsilon \downarrow 0$, the celebrated Levy-Lieb \cite{Lev79,Lie83} constrained search density functional 
\begin{align*}
    F_\varepsilon(\rho) \coloneqq \min \left\{ \int_{\R^n} \varepsilon |\nabla \Psi(x)|^2 dx + \int_{\R^n} c_n(x) |\Psi(x)|^2 \mathrm{d} X : \Psi \in \wedge^n L^2(\R) \quad \rho_\Psi = n\rho\right\},
\end{align*}
where $\rho_\Psi$ denotes the single-particle density of $\Psi$, converges to the multimarginal optimal transport problem~\eqref{eq:MMOT}  \cite{BDG12,CFK13,BD17,CFK18,Lew18}, which is also called the strictly correlated electrons (SCE) functional in the physics literature. Consequently, for one-dimensional systems, the Seidl (transport) map $T$ allow for an explicit description of the asymptotic behavior of the probability distribution associated to the minimizer of $F_\varepsilon(\rho)$ in the regime of strong interactions. While restricted to the one-dimensional case, this explicit construction of optimal maps has served as a fruitful test ground for the development of new density functionals aiming to capture the physics of general strongly correlated systems  \cite{RSG11,MMC+13,FGG22,VGD+23}. Moreover, this result also plays a central role in a recent derivation of the next order expansion of $F_\epsilon(\rho)$ in the semiclassical limit $\varepsilon \downarrow 0$. \cite{CDS25}.

However, the assumptions on the interaction potential in \cite{CDD15} are too restrictive for some applications \cite{VKS+04,FP05,LG12,LG13,CKG+17,PKF19}. To be more precise, these works deal with periodic systems, for which the natural interaction potentials must also be periodic, as particles are either restricted to the flat torus or a quantum ring. In particular, physically relevant interactions can not be strictly decreasing and the results in \cite{CDD15} do not immediately apply. It is therefore natural to ask the following question:
\begin{itemize}
    \item Can one extend the results in \cite{CDD15} to a larger class of interaction potentials? Better yet, can one characterize the maximal class of (not necessarily translation invariant) two-body interactions, for which~\eqref{eq:MMOT} always admit a minimizer of the form~\eqref{eq:optimizer}?
\end{itemize}
Answering this question is the first goal of this paper. 

The second goal of the paper is related to recent advances towards a rigorous mathematical foundation of density functional theory for one-dimensional systems \cite{SPR+24,Cor25a,Cor25b,Cor25c,SPR+25,CL25}. More precisely, in these works the authors show that, for any density function $\rho$ with finite kinetic energy that is strictly positive on an interval, and for a rather general class of pairwise interactions $w$, there exists an external potential $v = v(\rho,w)$ in the dual Sobolev space $\mH^{-1}$ such that $\rho$ is the ground-state density of the Hamiltonian 
\begin{align}
    H_n(v,w) = -\Delta + \sum_{i\neq j}^n w(x_i,x_j) - \sum_{j=1}^n v(x_j),\quad \mbox{acting on}\quad \mathcal{H}_n = \wedge^n \mL^2([0,2\pi])\label{eq:Hamiltonian}
\end{align}
under Neumann or periodic boundary conditions\footnote{We remark that the potential depends on the boundary conditions. Moreover, in the periodic case, $\rho$ may be only ensemble $v$-representable, i.e., the ground-state with single-particle density $\rho$ may be a mixed state.}. In particular, for such densities, these results guarantee the existence of the adiabatic connection, i.e., a map $\lambda \in \R \mapsto v_\lambda = v(\rho,\lambda w)$ such that $\rho$ is the ground-state density of the Schr\"odinger operator $H_n(v_\lambda,\lambda w)$ for every $\lambda \in \R$. Furthermore, in a recent work by the author and Laestadius \cite{CL25}, it is shown that the map $\lambda \mapsto v_\lambda$ is real analytic, thereby justifying the G\"orling-Levy perturbation series expansion of $v_\lambda$ in the weakly interacting limit $\lambda \rightarrow 0$. However, these results are restricted to the finite $\lambda$ case and does not provide information on the opposite --strongly interacting-- limit $\lambda \rightarrow \infty$. Therefore, in the current paper, our second main goal is to rigorously investigate the asymptotic behaviour of $v_\lambda$ in the limit $\lambda \rightarrow \infty$.

\subsection{Main contributions}
 In summary, the main contributions of the paper can be described as follows.
 \begin{itemize}
     \item We characterize the set of all pairwise interactions $w$ for which, for any $n\in \N$ and any $\rho \in \mathcal{P}(I)$, the Seidl plan is an optimizer of the $n$-marginal optimal transport problem with marginal $\rho$.
     \item We show that, for periodic systems in arbitrary dimensions, the Lieb functional converges to the optimal transport functional in the semiclassical regime $\varepsilon \downarrow 0$. 
     \item We show that, for one-dimensional periodic systems, the adiabatic potential converges towards the Kantorovich potential of the optimal transport problem in the limit $\varepsilon \downarrow 0$ (or equivalently $\lambda \uparrow \infty$). 
 \end{itemize}
\section{Main results}

We now turn to the precise statement of our main results. 
\subsection{Well-ordering costs and the Seidl conjecture} Our first main result gives a characterization of pair interactions for which the Seidl conjecture holds. To state it precisely, let us introduce the following definition.

\begin{definition}[Well-ordering interaction] \label{def:well-ordering} We say that a symmetric continuous function $w:J\times J \rightarrow \R\cup \{\infty\}$ is well-ordering in a set $J \subset \R$ if the following holds. For any $x_1\leq x_2 \leq x_3 \leq x_4 \in J$ we have
\begin{align*}
    w(x_1,x_3) + w(x_2,x_4) = \min \{ w(x_{\sigma(1)},x_{\sigma(2)}) + w(x_{\sigma(3)},x_{\sigma(4)}) : \sigma : \{1,2,3,4\} \rightarrow \{1,2,3,4\} \quad\mbox{bijective} \}.
\end{align*}
Moreover, we say that $w$ is strictly well-ordering if the equality
\begin{align*}
    w(x_1,x_3) + w(x_2,x_4) = w(x_{\sigma(1)},x_{\sigma(2)}) + w(x_{\sigma(3)},x_{\sigma(4)})
\end{align*}
holds true if and only if either $w(x_1,x_3) + w(x_2,x_4) = \infty$ or 
\begin{align*}
    \delta_{x_1} + \delta_{x_3} = \delta_{x_{\sigma(1)}} + \delta_{x_{\sigma(2)}} \quad \mbox{or} \quad \delta_{x_1} + \delta_{x_3} = \delta_{x_{\sigma(3)}} + \delta_{x_{\sigma(4)}},
\end{align*}
where $\delta_x$ denotes the Dirac delta measure at $x\in \R$.
\end{definition}

Using this definition, our first main result can be stated as follows.

\begin{theorem}[Optimal transport characterization of well-ordering interactions] \label{thm:main} Let $w:I \times I \rightarrow \R\cup \{+\infty\}$ be a continuous symmetric function on an closed interval $I=[a,b] \subset \R$ such that $F_{\rm OT}(\rho) < \infty$ for any non-atomic $\rho \in \mathcal{P}(I)$. Then, the measure $\gamma_\rho$ in~\eqref{eq:optimizer} is a minimizer of the MMOT problem~\eqref{eq:MMOT} for arbitrary $n\in \N$ and non-atomic $\rho \in \mathcal{P}(I)$ if and only if $w$ is well-ordering. Moreover, if $w$ is strictly well-ordering, then the symmetrization of $\gamma_\rho$ is the unique symmetric minimizer of~\eqref{eq:MMOT}.
\end{theorem}

It is not difficult to see that the well-ordering property is necessary and sufficient for $\gamma_\rho$ to be a minimizer of~\eqref{eq:MMOT} in the two marginal case. The striking feature of Theorem~\ref{thm:main} is that the well-ordering property, which is a two marginal condition, is also sufficient for the multimarginal case with an arbitrary number of marginals.

\begin{remark}[Only symmetric interactions matter] There is no loss of generality in assuming that $w$ is symmetric. Indeed, since the sum in the cost function~\eqref{eq:cost function} is taken with respect to all $i\neq j$, the cost is the same if we replace $w$ by its symmetric part $w_{\rm sym}(x,y) = \frac12 (w(x,y) + w(y,x))$. In particular, for translation invariant interactions $w(x,y) = w(x-y)$, it suffices to work with even functions.
\end{remark}

At a first glance, the well-ordering property may seem difficult to verify in practice. However, as we illustrate next with several examples, this is not the case.

\subsection*{Applications to multimarginal optimal transport on one-dimensional manifolds} Let us first consider the case of translation invariant interactions, $w(x,y) = w(x-y)$. In this case, one has the following reformulation of the well-ordering property.
\begin{proposition} Let $a<b$ and $w:[a-b,b-a]\rightarrow \R\cup \{+\infty\}$ be an even continuous function. Then $w(x-y)$ is well-ordering in $[a,b]$ if and only if the following holds:
\begin{enumerate}[label=(\roman*)]
    \item $w$ is convex, and
    \item $w$ satisfies 
    \begin{align}
        w(d_0 + \delta) + w(d_1+\delta) \leq w(d_1) + w(d_0), \quad \mbox{for any $d_0,d_1,\delta>0$ with $d_0 + d_1 + \delta \leq b-a$.} \label{eq:increasing condition}
    \end{align}
\end{enumerate}
Moreover, $w$ is strictly well-ordering if and only if $w$ is strictly convex and one has strict inequality in~\eqref{eq:increasing condition} for $\delta>0$.
\end{proposition}

We can now use the above reformulation to verify the well-ordering property in several cases that are physically relevant.
\begin{enumerate}[label=(\arabic*)]
\item (Unbounded intervals) First, in the case of an unbounded interval, i.e., $a = -\infty$ or $b = +\infty$, one can set $d_1 = d_0$ in~\eqref{eq:increasing condition} to see that $w$ must be decreasing. This shows that, for translation invariant interactions on unbounded intervals, the convex plus decreasing assumption on $w$ used in~\cite{CDD15} is optimal for the Seidl conjecture to hold.

\item (Flat torus) In the case of bounded domains, the decreasing condition is no longer necessary because inequality~\eqref{eq:increasing condition} (with $d_0 = d_1$) only needs to hold for $\delta \leq b-a - 2d_0$. For instance, one can show that interactions of the form
\begin{align}
    w(x,y) = w(x-y) = w(\min\{|x-y|,|x-y+2\pi|,|x-y-2\pi|\}) = w(|x-y|_{\mathbb{T}}), \label{eq:torus cost}
\end{align}
which are natural on the flat torus $\T = \R /(2\pi\Z)$, also satisfy~\eqref{eq:increasing condition}, provided that $w$ is convex and decreasing in $[0,\pi]$. Indeed, in this case, for $d_0+\delta, d_1+\delta < \pi $,~\eqref{eq:increasing condition} is immediate from the decreasing property of $w$, while for $d_0 +\delta >\pi$ we have
\begin{align*}
    w(d_0+\delta) + w(d_1+\delta) = w(2\pi-d_0-\delta) + w(d_1+\delta) \leq w(2\pi-d_0) + w(d_1) = w(d_0) + w(d_1),
\end{align*}
where we used that $w(x) = w(2\pi-x)$ for the equalities and the convexity of $w$ for the inequality. 
\item (Quantum ring) Similar considerations show that interactions of the form 
\begin{align}
    w(\theta,\phi) = w(2 \sin(|\theta-\phi|/2)) \label{eq:ring cost}
\end{align}
for $\theta,\phi\in [0,2\pi]$ are well-ordering in $I = [0,2\pi]$ if and only if $w$ is convex and decreasing in $[0,2]$. Note that these are the physically relevant interactions in the case of particles leaving in the ring $S^1 = \{x\in \R^2 : |x| =1 \}$. Indeed, in this case, the distance between two particles at positions $x_1 = \rm{e}^{i \theta_1} \in S^1$ and $x_2 = \rm{e}^{i \theta_2}\in S^1$ is given by $2 \sin(|\theta_1-\theta_2|/2)$.
\end{enumerate}
To illustrate that Theorem~\ref{thm:main} is applicable beyond the case of translation invariant interactions, we also consider the following examples. A proof is presented in the appendix.
\begin{enumerate}[label=(\arabic*),start=4]
\item (Trivial interaction) Any function of the form $w(x,y) = f(x) + f(y)$ for any $f$ is trivially well-ordering. Of course, such examples are not interesting as $w$ is an one-body operator and not a real pairwise interaction. 
\item (Particles on a graph) Let $f:[0,\infty) \rightarrow \R$ and $g:[0,\infty) \rightarrow \R\cup\{+\infty\}$ be convex and non-increasing functions, then the interaction
\begin{align}
    w(x,y) = g\left(\sqrt{(x-y)^2 + (f(x)-f(y))^2}\right) \label{eq:graph example}
\end{align}
is well-ordering in $[0,\infty)$. Interactions of this form are natural for particles confined to the graph $\mathrm{Gr}(f) =\{(x,f(x)) : x\in [0,\infty)\}\subset \R^2$. 
\item (Cone of well-ordering interactions) It is easy to verify that the space of well-ordering interactions is a cone, i.e., closed under pointwise addition and multiplication by positive constants. In particular, interactions of the form
\begin{align*}
    w(x,y) = \sum_{j=1}^m w_j(x,y) \quad
\end{align*}
with $w_j$ well-ordering are also well-ordering. Moreover, $w$ is translation invariant if and only if every $w_j$ is translation invariant. 
\end{enumerate}

\subsection{Strictly correlated electrons on a quantum ring} We now turn to the results concerning the strongly interacting limit of the Lieb functional in a quantum ring. To state it precisely, let us introduce the Lieb density functional in the periodic setting as follows.
\begin{align}
    F_{\rm per}^\varepsilon(\rho) \coloneqq \inf \left\{ \sum_{k=1}^\infty \lambda_k  \int_{I_n} |\nabla \Psi _k(x)|^2 + c_n(x) |\Psi_k(x)|^2 \mathrm{d} x: \Psi_k \in \mathcal{H}_n \cap \mH^1_{\rm per}(I_n), \quad \sum_{k} \lambda_k \rho_{\Psi_k} = n\rho \right\}, \label{eq:Lieb functional}
\end{align}
where $\mathcal{H}_n = \wedge^n \mL^2(I)$ and $\mH^1_{\rm per}(I_n)$ denotes the space of Sobolev functions on the box $I_n = [0,2\pi]^n$ with periodic boundary conditions. Notice that the domain of $F_{\rm per}^\varepsilon$ is contained in $\mH^1_{\rm per}(I)$ because any periodic wavefunction with finite kinetic energy has single-particle density in $\sqrt{\rho} \in \mH^1_{\rm per}(I)$. We also emphasize that $F^\varepsilon_{\rm per}$ depends on the number of electrons $n$, though this dependence will be omitted in the notation for simplicity.

Let us also impose the following assumption on the pairwise interaction. 
\begin{assumption} \label{assump:interaction} Let $w:I\times I \rightarrow \R_+ \cup \{+\infty\}$ be a symmetric continuous\footnote{By continuity in $\R\cup \{+\infty\}$, we mean with respect to the topology in $\R\cup\{+\infty\}$ generated by the intervals $(a,b)$ and $(a,\infty]$.} function, then we assume that
\begin{align}
    \{ w= \infty\} = D_0 \coloneqq  \{ (x,y) \in I\times I : |x-y|_{\mathbb{T}} = 0\}, \label{eq:periodic diagonal}
\end{align}
where $|\cdot|_{\T}$ is the torus norm defined in~\eqref{eq:torus cost}. Note that this assumption holds for any costs of the form~\eqref{eq:torus cost} and~\eqref{eq:ring cost} with continuous $w\geq 0$ such that $\lim_{t\downarrow 0} w(t) = \infty$.
\end{assumption}

\begin{remark}[Non-negative interactions] There is no loss of generality in assuming that $w \geq 0$. Indeed, since $I\times I$ is compact, any continuous function $w:I\times I\rightarrow \R \cup\{+\infty\}$ is bounded from below; therefore, we can simply shift the cost by a sufficiently large positive constant.
\end{remark}
Under this assumption, we have the following result on the periodic Lieb functional. 

\begin{theorem}[Strictly correlated electrons on a quantum ring] \label{thm:SCE limit} Let $I = [0, 2\pi]$ and $\rho \in \mathcal{P}(I)$ be such that $\sqrt{\rho} \in \mH^1_{\rm per}(I)$. Suppose that $w$ satisfies Assumption~\ref{assump:interaction}. Then 
\begin{align}
    \lim_{\varepsilon \downarrow 0} F_{\rm per}^\varepsilon(\rho) = F_{\rm OT}(\rho). \label{eq:SCE limit}
\end{align}
Moreover, up to subsequences, the optimizers $\Psi_\varepsilon$ of $F_{\rm per}^\varepsilon(\rho)$ satisfy $|\Psi_\varepsilon|^2 \rightharpoonup \gamma_{\rm opt}$ as $\varepsilon \downarrow 0$ in the sense of weak convergence of measures, where $\gamma_{\rm opt}$ is a symmetric optimizer of~\eqref{eq:MMOT}. In particular, if $w$ is strictly well-ordering, then $|\Psi_\varepsilon|^2 \rightharpoonup \gamma_\rho$, where $\gamma_\rho$ is the symmetrization of the Seidl optimizer~\eqref{eq:optimizer}.
\end{theorem}

Let us now briefly comment on the connection of Theorem~\ref{thm:SCE limit} with previous results in the literature. First, we note that, in the whole space ($\R^d$ with arbitrary $d\in \N$) setting, this result was anticipated by Seidl and co-workers in the physics literature \cite{Sei99,SPL98,SGS07}, and first rigorously derived for $n=2$ electrons in the work by Cotar et al \cite{CFK13}. Later this result was extended to $n=3$ in \cite{BD17}, and finally to any $n\in \N$ by Lewin \cite{Lew18} in the mixed state setting, and by Cotar et al \cite{CFK18} in the pure-state setting. Nevertheless, to the best of the author's knowledge, none of the previous works deal with the periodic setting; in particular, Theorem~\ref{thm:SCE limit} appears to be new. Moreover, our proof relies on an adaptation of the Lewin construction \cite{Lew18} (which is based on the regularization procedure by Bindini and De Pascale \cite{BD17}) to the periodic setting. Consequently, this result can be extended to periodic systems in higher dimensions, as shown later in Theorem~\ref{thm:SCE limit higher dimensions}.

\begin{remark}[Neumann and Dirichlet cases] It is interesting to note that, by considering the periodic case, one can also establish the convergence in~\eqref{eq:SCE limit} for the functionals in the Dirichlet case and in the Neumann case with periodic densities. More precisely, if we define $F^\varepsilon_{\rm D/N}(\rho)$ as in~\eqref{eq:Lieb functional} but with wavefunctions in $\mH^1_0(I_n) \cap \mathcal{H}_N$, respectively, $\mH^1(I_n) \cap \mathcal{H}_N$ (instead of $\mH^1_{\rm per}(I_n) \cap \mathcal{H}_N$), then it is immediate to see that
\begin{align*}
    \begin{dcases} 
    F_{\rm OT}(\rho) \leq F^\varepsilon_{\rm D}(\rho) = F_{\rm per}^\varepsilon(\rho), \quad &\mbox{for $\sqrt{\rho} \in \mH^1_0(I)$, and} \\
    F_{\rm OT}(\rho) \leq F^\varepsilon_{\rm N}(\rho) \leq F_{\rm per}^\varepsilon(\rho) \quad &\mbox{for $\sqrt{\rho} \in \mH^1_{\rm per}(I)$. }\end{dcases}
\end{align*}
Indeed, the second line follows from the obvious fact that $\mH^1_{\rm per}(I_n)\cap \mathcal{H}_n \subset \mH^1(I_n) \cap \mathcal{H}_n$, while the first line follows from the fact that any (periodic) wave-function with density in $\mH^1_0(I)$ must vanish along the boundary $\partial I_n$. Hence, we have
\begin{align*}
    \lim_{\varepsilon \downarrow 0} F_{D/N}^\varepsilon(\rho) = F_{\rm OT}(\rho) \quad \mbox{(with the restriction to periodic densities in the Neumann case).}
\end{align*}
\end{remark}

\begin{remark}[Levy-Lieb functional] Let us also mention that, for an odd number of particles in one dimension, Theorem~\ref{thm:SCE limit} also holds for the Levy-Lieb (constrained search) functional
\begin{align*}
    F_{LL,\rm per}^\varepsilon(\rho) \coloneqq \inf \left\{ \varepsilon T(\Psi) + \int c_n(x) |\Psi(x)|^2 \mathrm{d} x : \Psi \in \mathcal{H}_n \cap \mH^1_{\rm per}(I_n) \quad \rho_\Psi = n \rho\right\}.
\end{align*}
The reason is that any symmetric probability density $\mathbb{P}$ with finite kinetic energy that is periodic and vanishes along the coalescence points, can be turned into the probability density of an anti-symmetric wave-function via the following Bose-Fermi map (see, e. g. \cite{Gir60,Cor25c}):
\begin{align*}
    \sqrt{\mathbb{P}} \in \mH^1_{\rm per}(I_N)\cap \otimes_{\rm sym}^n \mL^2(I)  \mapsto \prod_{1=i<j<n} \mathrm{sign}(x_i-x_j) \sqrt{\mathbb{P}(x_1,...,x_n)} \in \mH^1_{\rm per}(I_n) \cap \mathcal{H}_n .
\end{align*}
In particular $F_{\rm per}^{\varepsilon}(\rho) = F_{LL,\rm per}^\varepsilon(\rho)$ for $n \in 2\N +1$. Note that, for an even number of particles, the right-hand side is no longer periodic, so the same argument does not work. On the other hand, the corresponding equality holds for the Dirichlet and Neumann case for any number of particles $n\in \N$ via the same argument.
\end{remark}

\subsection{From Kohn-Sham to Kantorovich potentials} We now present our main result concerning the strongly interacting limit of the adiabatic potential. To this end, let us first recall that the set
\begin{align}
    \mathcal{D}_{\rm per} \coloneqq \left\{ \rho \in \mH^1_{\rm per}(I) : \rho(x)>0 \mbox{ for any $x \in I = [0,2\pi]$} \quad \int_I \rho(x) = 1 \right\} \label{eq:periodic v-rep densities}
\end{align}
is contained in the set of ensemble $v$-representable densities on the torus $\mathbb{T}$ for general interactions (see \cite{SPR+24}), and coincides with the set of non-interacting $v$-representable densities (see \cite{Cor25a}). More precisely, one can show (see Lemma~\ref{lem:existence of potential}) that, under Assumption~\ref{assump:interaction} on the interaction, for any $\rho \in \mathcal{D}_{\rm per}$, there exists $ v_\varepsilon(\rho) \in \mH^{-1}_{\rm per}(I)$ such that 
\begin{align}
    \min \left\{E_{v_\varepsilon(\rho)}(\zeta) \coloneqq F^\varepsilon_{\rm per}(\zeta) - n\inner{v_\varepsilon(\rho),\zeta} : \sqrt{\zeta} \in \mH^1_{\rm per}(I) ,\quad \int_I \zeta = 1 \right\} = F^\varepsilon_{\rm per}(\rho) - n\inner{v_\varepsilon(\rho),\rho} .\label{eq:potential duality}
\end{align}
In other words, $\rho$ is the minimizing (or ground-state) density for the energy of the Hamiltonian $H_n(v_\varepsilon(\rho)/\varepsilon,w/\varepsilon)$ formally introduced in~\eqref{eq:Hamiltonian}.

In the optimal transport case, the analogous potential is the so-called Kantorovich potential, whose precise definition we recall next.
\begin{definition}[Kantorovich potentials]\label{def:Kantorovich potentials} We say that a function $v_{\rm OT} = v_{\rm OT}(\rho):I \rightarrow \R$ is a Kantorovich potential for the optimal transport problem with marginal $\rho \in \mathcal{P}(I)$, if it is $\rho$-integrable and satisfies
\begin{align*}
    F_{\rm OT}(\rho) = \int_I \rho(x) v_{\rm OT}(x)\mathrm{d} x \quad \mbox{and}\quad c_n(x_1,...,x_n) - \sum_{j=1}^n v_{\rm OT}(x_j) \geq 0, \quad \mbox{for every $(x_1,...,x_n) \in I_n$.}
\end{align*}
Moreover, we shall say that $v_{\rm OT}(\rho)$ is a regular Kantorovich potential if it is continuous.
\end{definition}
We then show that, in the limit as $\varepsilon \downarrow 0$, the potential $v_\varepsilon(\rho)$ converges to a regular Kantorovich potential. To the best of our knowledge, this is the first rigorous justification of the asymptotic expansion of the adiabatic potential in the strongly interacting limit that often appears in the physics literature (see Remark~\ref{rem:asymptotics adiabatic potential} below).

\begin{theorem}[Asymptotics of the potential in the semiclassical limit] \label{thm:Kantorovich potentials} Let $\rho \in \mathcal{D}_{\rm per}$ and $v_\varepsilon(\rho) \in \mH^{-1}_{\rm per}(I)$ be any potential such that~\eqref{eq:potential duality} holds. Then, up to subsequences, there exists $\{C_\varepsilon \} \subset \R$ such that
\begin{align*}
    v_{\varepsilon}(\rho) + C_\varepsilon\rightharpoonup v_{\rm OT}(\rho) \quad \mbox{in $\mH^{-1}_{\rm per}(I)$, as $\varepsilon \downarrow 0$,}
\end{align*}
for some regular Kantorovich potential $v_{\rm OT}(\rho)$. Moreover, if $w$ is locally $C^1$ (away from the diagonal), then the potential $v_{\rm OT}(\rho)$ is unique and the convergence holds without appealing to subsequences.
\end{theorem}

\begin{remark}[Normalization] The constants $C_\varepsilon$ in Theorem~\ref{thm:Kantorovich potentials} can be chosen as $C_\varepsilon := \frac{1}{n} F_\varepsilon(\rho) -\inner{v_\varepsilon(\rho),\rho}$.
\end{remark}

\begin{remark}[Strongly interacting asymptotics] \label{rem:asymptotics adiabatic potential} In the notation previously used for the adiabatic connection, the result of Theorem~\ref{thm:Kantorovich potentials} can be stated as
\begin{align*}
    \lim_{\lambda \rightarrow \infty} \frac{v_\lambda(\rho)}{\lambda} = v_{\rm OT}(\rho), \quad \mbox{or} \quad v_\lambda(\rho) = \lambda v_{\rm OT}(\rho) + o(\lambda),
\end{align*}
which is the asymptotic expansion often appearing in the physics literature \cite{GSG+19}. It would be interesting to rigorously obtain the next-order correction for the potential $v_\lambda(\rho)$, which is conjectured to come from the zero-point oscillations functional \cite{Sei99,GSG+19}. As remarked in \cite{GSG+19}, this seems to be a necessary step to go beyond the two-term energy asymptotics, which was rigorously established under different assumptions for 1D systems in \cite{CDS25}.
\end{remark}

\subsection{Outline of the proof and structure of the paper}

We now briefly outline the main steps in the proofs of our main theorems and how these steps are organized throughout the paper.

In Section~\ref{sec:well-ordering}, we present the proof of Theorem~\ref{thm:main}. As in \cite{CDD15}, the key step of the proof is a characterization of $c$-cyclically monotone sets for an arbitrary number of marginals $n$, see Proposition~\ref{prop:well-ordering}. However, in contrast to \cite{CDD15}, we do not assume the interaction potential to be the decreasing, which is crucial for their main estimate in \cite[Lemma 3.4]{CDD15}. In fact, it is not difficult to show that this lemma no longer holds in the general case considered here. Therefore, our strategy here is considerably different; it relies on an explicit algorithmic procedure to reduce the sum of the costs of a balanced bi-partition of $2n$ points, by swapping suitably chosen pairs of consecutive points (see Section~\ref{sec:c monotone}). The choice of the swapping pairs is made by carefully analyzing an auxiliary function introduced later (see Lemma~\ref{lem:elementary}). We can then show that this procedure only terminates when the balanced bi-partition is well-ordered (see Lemma~\ref{lem:swap}). This strategy is more general than the previous approach and one of the main novelties of the paper. Once the geometric characterization of $c$-cyclically monotonicity is established, the rest of the proof follows the same arguments as in \cite{CDD15}. 

In Section~\ref{sec:SCE limit}, we prove Theorem~\ref{thm:SCE limit}. This proof relies on two main steps. In the first step we show that, under Assumption~\ref{assump:interaction}, optimal plans are supported away from the periodic set of coalescence points (see Lemma~\ref{lem:support}). This result extends previous results from \cite{BCD18,CDS19} to the periodic setting and the proof closely follows their arguments. The second step is presented in Section~\ref{sec:Lewin construction} and consists in a simple adaptation of the regularization procedure from \cite{Lew18} to the periodic setting. We can then combine these two steps with standard arguments to complete the proof of Theorem~\ref{thm:SCE limit}. As previously emphasized, this strategy also extends to higher dimensions. Moreover, under additional regularity assumptions on $w$, it also allows us to obtain an estimate of order $\varepsilon^{\frac12}$ for the remainder, see Theorem~\ref{thm:SCE limit higher dimensions} below. 

The proof of Theorem~\ref{thm:Kantorovich potentials} is carried out in Section~\ref{sec:potentials} and consists in three main steps. In the first step, we show the existence of generalized Kantorovich potentials in the dual space of $\mH^1_{\rm per}$, see Lemma~\ref{lem:existence of potential}. This step is presented in Section~\ref{sec:existence} and relies on a well-known result in convex analysis and the simple but important observation that the set $\mathcal{D}_{\rm per}$ is the relative interior of the set of representable densities $\mathcal{R}_{\rm per}$ (whose definition is recalled later in~\eqref{eq:representable densities}) with respect to the $\mH^1$ norm. This observation was used in several recent works by the author and others \cite{SPR+24,Cor25c,SPR+25,CL25}, and is the main reason for the restriction to one-dimensional systems. In the second step of the proof, which is conducted in Section~\ref{sec:regularity}, we show that for bounded and continuous cost, any generalized Kantorovich potential is in fact a classical continuous Kantorovich potential. The main tool at this step is the Riesz representation theorem for non-negative functionals and a weak version of the multimarginal $c$-transform. The third and last step of the proof consists in showing that, locally in $\mathcal{D}_{\rm per}$, the optimal transport problem with unbounded cost is equivalent to the problem with truncated costs. This step is conducted in Section~\ref{sec:reduction to truncated cost} and also relies on the openess of $\mathcal{D}_{\rm per}$ and on extensions of previous results from \cite{BCD18,CDS19}. 

\section{Well-ordering interactions and the Seidl conjecture}
\label{sec:well-ordering}

Our goal in this section is to prove Theorem~\ref{thm:main}. The main step in the proof is the following geometric characterization of $c$-cyclically monotone sets.
\begin{proposition}[Well-ordering and $c$-monotonicity] \label{prop:well-ordering} 
Suppose that $w$ is well-ordering in an interval $I\subset \R$, $n\in \N$, and let $c_n$ be defined as
\begin{align*}
    c_n(x) = c_n(x_1,...,x_n) = \sum_{i\neq j} w(x_i,x_j), \quad \mbox{for $x = (x_1,...,x_n) \in I^n$.}
\end{align*}
Let $x_1 \leq x_2 ... \leq x_{2n} \in I$, then for any $A\subset \{1,...,2n\}$ such that $|A| = n$, we have
\begin{align}
    c_n(x_{A_o}) + c_n(x_{A_e}) \leq c_n(x_{A}) + c_n(x_{A^c}), \label{eq:n well-ordering}
\end{align}
where $x_A = (x_{a_1},...,x_{a_n})$, $A^c = \{1,...,2n\} \setminus A$ is the complementary set of $A$, and $A_o = \{1,3,...,2n-1\}$ and $A_e=\{2,4,6,...,2n\}$ are respectively the odd and even numbers in $\{1,2,...,2n\}$. Moreover, if $c(x_A)+c(x_{A^c}) < \infty$ and $w$ is strictly well-ordering, then
\begin{align*}
    c_n(x_A) + c_n(x_{A^c}) = \min \{ c_n(x_B) + c_n(x_{B^c}) : B \subset \{1,...,2n\},\quad |B| = n \}
\end{align*}
if and only if
\begin{align*}
    \sum_{k \in A} \delta_{x_k} = \sum_{k = 1}^n \delta_{x_{2k}} \quad \mbox{or} \quad \sum_{k\in A} \delta_{x_k} = \sum_{k =1}^n \delta_{x_{2k-1}}.
\end{align*}
\end{proposition}

Proposition~\ref{prop:well-ordering} is a generalization of \cite[Proposition 2.4]{CDD15}. However, as previously noted, their proof crucially relies on an estimate for $\ell$ neighbors, which uses the decreasing property of the interaction potential in a critical way. Unfortunately, this property is no longer available in the general setting considered here; in fact, it is not difficult to see\footnote{For instance, a simple counterexample to the $\ell$ neighbors estimate is the following: if we consider $w$ as in \eqref{eq:ring cost}, $n=3$ and the points $0 \approx x_1 = x_2 < x_3 = x_1 + \delta < x_4 < x_5-\delta < x_5 = x_6 \approx 2\pi$ and set $A = \{1,2,3\}$ and $A^c = \{4,5,6\}$. Then the sum of the $\ell =2$  neighbors for this configuration is
\begin{align*}
    w(|x_1-x_3|) + w(|x_4-x_6|) \approx 2w(\delta)
\end{align*}
while the sum of the $2$ neighbors for the optimal configuration $A_e = \{1, 3, 5\}, A_o = \{2, 4, 6\}$ is
\begin{align*}
    w(x_1,x_5) + w(x_2,x_6)\approx 2w(2\pi) > 2 w(\delta)
\end{align*}} that the $\ell$ neighbors estimate in \cite[Lemma 3.4]{CDD15} can no longer hold true for general interactions. Therefore, we need to adopt a different strategy, which seems somewhat more fundamental than the previous approach.

\subsection{Geometric characterization of $c$-monotonicity} \label{sec:c monotone} Our new strategy relies on the following auxiliary function. Let $A\subset \{1,...,2n\}$ be a subset with $n$ elements, and define the measure
\begin{align*}
    \mu_A = \sum_{j\in A} \delta_j - \sum_{j A^c} \delta_j,
\end{align*}
where $A^c := \{1,...,2n\} \setminus A$ is the complementary set of $A$ in $\{1,...,2n\}$. We then define the function $f_A(t)$ as the cumulative function of $\mu_A$, i.e.,
\begin{align}
    f_A(t) = \int_0^t \mu_A(ds). \label{eq:f def}
\end{align}
The next lemma summarizes a few elementary properties of the function $f_A$.

\begin{lemma}[Elementary properties of $f$] \label{lem:elementary} Let $A \subset \{1,..,2n\}$ with $|A| = n$ and let $f_A$ be defined as in~\eqref{eq:f def}. Then the following holds:
\begin{enumerate}[label=(\arabic*)] 
\item \label{it:jump property} The function $f_A$ is integer valued, constant on intervals of the form $[n,n+1)$, and has jumps of size $1$ at the points $\{1,...,2n\}$. Moreover, $f_A(t) = 0$ for any $t< 1$ or $t\geq 2n$.
\item \label{it:exchange property} We have $f_{A^c} = - f_A$.
\item \label{it:oscillation property} We have $A = A_e$ or $A= A_o$ if and only if the function $f$ satisfies
\begin{align*}
    \rm{Osc}(f) = \max f_A - \min f_A = 1.
\end{align*}
\end{enumerate}
\end{lemma}

\begin{proof} The proof is straightforward from the definition of $f_A$. \end{proof}

The main idea of the proof now is to show that, starting with any set of indices $A\subset \{1,...,2n\}$ with $|A| = n$ we can exchange points between $A$ and $A^c$ to construct a new set of indices $B$ such that the cost decreases (or at least does not increase) and the oscillation of $f_{B}$ is strictly smaller than $f_A$. For this, we shall use the following key observation.

\begin{lemma}[Partition via maximum points] \label{lem:maximum points} Let $\{\ell_j\}_{j\leq m} \in \{1,2,...2n\}$ be the integer maximum points of $f_A$ ordered increasingly. Without loss of generality, we assume $\ell_m \leq 2n-1$, as otherwise, we can work with $A^c$ instead. Then $\ell_j+1 \in A^c$ for any $j\leq m$ and there exists a bijective map $\sigma:A^c \setminus \{\ell_j+1\}_{j\leq m}\rightarrow A \setminus \{\ell_{j}\}_{j\leq m}$ such that 
\begin{enumerate}[label=(\roman*)]
\item \label{it:increasing v1} For any $1\leq j \leq m-1$, the restriction of $\sigma$ to $(\ell_j+1,\ell_{j+1})$ maps $A^c \cap(\ell_{j}+1,\ell_j)$ to $A\cap (\ell_j +1, \ell_{j+1})$ and satisfies $\sigma(k) > k$ for any $k\in A^c (\ell_j+1,\ell_j)$.
\item \label{it:increasing v2} The restriction of $\sigma$ to $M:= [1,\ell_1) \cup (\ell_m,2n]$ maps $A^c \cap M$ to $A\cap M$ and satisfies $\sigma(k) > k$ if we identify $(\ell_m+1,2n]$ with $(\ell_{m}+1-2n,0]$, i.e., the function $\tau^{-1} \circ \sigma \circ \tau :\tau^{-1}(A^c) \cap (\ell_m+1-2n,\ell_1)\rightarrow \tau(A) \cap (\ell_m-2n,\ell_1)$, where
\begin{align*}
    \tau : (\ell_m+1-2n, \ell_m] \rightarrow (0,2n],\quad \tau(k) = \begin{dcases} 
    k+2n, \quad &\mbox{if  $\ell_m+1 - 2n < k\leq 0$,} \\
    k, \quad &\mbox{if $0< k\leq \ell_m+1$,}  \end{dcases}
\end{align*}
is bijective and satisfy $\tau^{-1}(\sigma(\tau(k))) > k$ for $k\in \tau^{-1}(A^c) \cap (\ell_m+1-2n,\ell_1)$.
\end{enumerate}
\end{lemma}

\begin{proof} We first claim that for any $t\in (\ell_{k}+1,\ell_{k+1})$ and $1\leq k \leq m-1$, we have
\begin{align}
    |(\ell_{k}+1,t]\cap A^c| - |(\ell_{k}+1,t] \cap A| \geq 0 . \label{eq:claimed ineq}
\end{align}
To prove this claim, first note that, since $\ell_{k}$ and $\ell_{k+1}$ are two consecutive global maxima of $f_A$, $f_A$ is integer valued, and $f_A$ is constant on intervals of the form $[n,n+1)$, we must have
\begin{align}
    f_A(t) = f_A(\ell_k+1) - |(\ell_{k}+1,t]\cap A^c| + |(\ell_{k}+1,t] \cap A| \leq f(\ell_k)-1, \quad \mbox{for any $t\in (\ell_k+1,\ell_{k+1})$.}\label{eq:prev obst}
\end{align}
Moreover, as $f$ has jumps of size $1$ at each integer, we must also have $f(\ell_k+1) = f(\ell_k)-1$ (as otherwise we would have $f(\ell_k+1) = f(\ell_k)+1$ contradicting the maximality of $\ell_k$). This observation together with~\eqref{eq:prev obst} then implies~\eqref{eq:claimed ineq}. 

A similar argument shows that $|(\ell_{k}+1,\ell_{k+1})\cap A^c| = |(\ell_{k}+1,\ell_{k+1}) \cap A|$. In particular we have an even number of points in the interval $(\ell_k+1,\ell_{k+1})$. One can now construct an increasing bijection $\sigma : A^c \cap (\ell_k+1,\ell_{k+1}) \rightarrow A\cap (\ell_{k}+1,\ell_{k+1})$ as follows. Denoting by $a_1 < a_2 ... < a_p$ and $b_1 < b_2 ... < b_p$ respectively the elements of $A \cap (\ell_{k}+1,\ell_{k+1})$ and $A^c \cap (\ell_k+1,\ell_{k+1})$, we set $\sigma(b_j) \coloneqq a_j$. Then clearly $\sigma : A^c \cap (\ell_k+1,\ell_{k+1}) \rightarrow A\cap (\ell_{k}+1,\ell_{k+1})$ is a bijection. Moreover, we also have $a_j > b_j$. Indeed, suppose this is not the case, i.e., $a_j < b_j$ (as $a_j \neq b_j$), then there exists $a_1 <... a_j < t < b_j ... < b_p$. In particular, $|A^c \cap (\ell_k+1,t]| - |A\cap (\ell_k+1,t]| \leq (j-1) - j = -1$ contradicting~\eqref{eq:claimed ineq}. 

To prove the second statement, we can simply extend $f_A$ by setting $\tilde{f}_A(t) = f_A(t-2n)$ for $t\geq 2n$ (note that $f_A(2n) = f_A(0) = 0$) and argue as before. 
\end{proof}

We can now modify $A$ by exchanging each maximum point of $f_A$ (which belong to $A$) with its consecutive point (which belong to $A^c$). Precisely, let $A \subset \{1,...,2n\}$ with $|A| = n$, and let $\ell_1<\ell_2 ... < \ell_m \in \{1,...,2n\}$ be the ordered integer maximum points of $f_A$. Without loss of generality, we can assume that $\max f_A \geq 1$, as otherwise we can work with $f_{A^c}$ instead. In particular, we can assume $\ell_m < 2n$. We now define $B(A)$ as
\begin{align}
    B(A) := A\cup \{\ell_j+1\}_{1\leq j \leq m} \setminus \{\ell_j\}_{1\leq j\leq m}. \label{eq:B def}
\end{align}
The next lemma then shows that the cost of $B(A)$ is not larger than the cost of $A$. A visual example of $A$, $B(A)$, $f_A$ and $f_{B(A)}$ is provided in Figure~\ref{fig:swap}.

\begin{lemma}[Swapping maximum points]\label{lem:swap} Let $A$ be as before and $B(A)$ defined via~\eqref{eq:B def}. Suppose that $\mathrm{Osc}(f_A) \geq 2$. Then we have $\mathrm{Osc}(f_{B(A)}) < \mathrm{Osc}(f_A)$. Moreover,  for any $x_1\leq x_2 ... \leq x_{2n} \in I$, we have 
\begin{align}
    c_n(x_{B(A)}) + c_n(x_{B(A)^c}) \leq c_n(x_A) + c_n(x_{A^c}). \label{eq:reducing cost}
\end{align}
Moreover, if $w$ is strictly well-ordering and $c_n(x_{B(A)}) + c_n(x_{B(A)^c}) < \infty$, then equality holds if and only if 
\begin{align*}
    \sum_{k\in A} \delta_{x_k} = \sum_{k\in B(A)} \delta_{x_k} \quad\mbox{or}\quad \sum_{k\in A} \delta_{x_k} = \sum_{k \in B(A)^c} \delta_{x_k}.
\end{align*}
\end{lemma}

\begin{proof} Let $C \coloneqq A \cap B(A) = A \setminus \{\ell_j\}_{j\leq m}$ and $D = A^c \cap B(A)^c = A^c \setminus \{\ell_j+1\}_{j\leq m}$. Then we have
\begin{align*}
    \mu_{B(A)} = \sum_{j \in C} \delta_j - \sum_{j \in D} \delta_j + \sum_{k=1}^m \delta_{\ell_{k}+1} - \delta_{\ell_k} = \mu_A + 2\sum_{k=1}^m  \delta_{\ell_k+1} -\delta_{\ell_k},
\end{align*}
and therefore
\begin{align*}
    f_{B(A)}(t) = f_A(t) - 2\sum_{k=1}^m \mathbb{1}_{[\ell_j,\ell_{j}+1)}(t).
\end{align*}
Consequently $f_{B(A)}(t) = \max f_A - 2$ for $t\in [\ell_k,\ell_k+1)$, which implies that $\max f_{B(A)} = \max\{ f_A(x) : x \in  \R \setminus \cup_{k=1}^m [\ell_k, \ell_k+1)\} \leq \max f_A -1$. On the other hand, as $f_{B(A)}(t) = f_A(t)$ for $t\in \R\setminus \cup_{k=1}^m [\ell_k \ell_k+1)$, we must have $\min f_{B(A)} = \min \{\min f_A, \max f_A -2\}$. Together, these two observations imply that
\begin{align*}
    \mathrm{Osc}(f_{B(A)}) \leq \max f_A -1 - \min \{ \min f_A, \max f_A-2\} \leq \mathrm{Osc}(f_A),
\end{align*}
with equality with and only if $\mathrm{Osc}(f_A) = 1$ (as $\min f_A \leq \max f_A -1$). This proves the first statement.

For the second statement, we shall use Lemma~\ref{lem:maximum points}. To this end, first notice that
\begin{multline*}
    c_n(x_{A}) + c_n(x_{A^c}) = \sum_{j\neq k \in C} w(x_j,x_k) + \sum_{k\neq j}^{m} w(x_{\ell_j},x_{\ell_k}) + \sum_{j\neq k \in D} w(x_j,x_k) + \sum_{j\neq k}^m w(x_{\ell_j+1},x_{\ell_k+1}) \\
    + 2 \sum_{k=1}^m \left(\sum_{j\in C} w(x_j,x_{\ell_k}) +  \sum_{j\in D} w(x_j,x_{\ell_k+1})\right)
\end{multline*}
and
\begin{multline*}
    c_n(x_{B(A)}) + c_n(x_{B(A)^c}) = \sum_{j\neq k \in C} w(x_j,x_k) + \sum_{k\neq j}^{m} w(x_{\ell_j},x_{\ell_k}) + \sum_{j\neq k \in D} w(x_j,x_k) + \sum_{j\neq k}^m w(x_{\ell_j+1},x_{\ell_k+1})\\
    + 2 \sum_{k=1}^m \left(\sum_{j\in C} w(x_j,x_{\ell_k+1}) +  \sum_{j\in D} w(x_j,x_{\ell_k})\right).
\end{multline*}
Therefore, comparing the two expressions, it suffices to show that
\begin{align*}
    \sum_{j\in C} w(x_j,x_{\ell_k+1}) +  \sum_{j\in D} w(x_j,x_{\ell_k}) \leq \sum_{j\in C} w(x_j,x_{\ell_k}) +  \sum_{j\in D} w(x_j,x_{\ell_k+1}), \quad \mbox{for any $1\leq k \leq m$.}
\end{align*}
For this, we note that the function $\sigma$ from Lemma~\ref{lem:maximum points} is a bijective map from $D = A^c\setminus \{\ell_j+1\}_{j\leq m}$ to $C = A\setminus \{\ell_j\}_{j\leq m}$. We now claim that
\begin{align}
    w(x_{\sigma(j)},x_{\ell_k+1}) + w(x_j,x_{\ell_k}) \leq w(x_{\sigma(j)},x_{\ell_k}) + w(x_j,x_{\ell_k+1}), \quad \mbox{for any $j\in D$ and $1\leq k \leq m$.} \label{eq:goal ineq}
\end{align}
Indeed, if $j> \ell_k+1$, then either $j\in (\ell_{p}+1,\ell_{p+1})$ for some $k\leq p \leq m-1$ or $j\in (\ell_m+1,2n]$. In the first case, we must have $\sigma(j)> j$ by property~\ref{it:increasing v1}. In the second case, property~\ref{it:increasing v2} implies that either $\sigma(j)>j$ or $\sigma(j) < \ell_k$. Either way, since $x_1\leq... \leq x_{2n}$, we must have that either 
\begin{align}
    x_{\ell_k} \leq x_{\ell_k+1} \leq x_j \leq x_{\sigma(j)}\quad \mbox{or} \quad x_{\sigma(j)}\leq x_{\ell_k} \leq x_{\ell_k+1} \leq x_j. \label{eq:correct ordering}
\end{align} 
Therefore, inequality~\eqref{eq:goal ineq} follows because $w$ is well-ordering. Similarly, if $j<\ell_k$, we must have $j<\sigma(j)<\ell_k$, and therefore 
\begin{align}
    x_j\leq x_{\sigma(j)} \leq x_{\ell_k} \leq x_{\ell_{k}+1}. \label{eq:correct ordering 2}
\end{align} Thus inequality~\eqref{eq:goal ineq} follows again from the well-ordering property of $w$.

Let us now carefully look into the equality case in~\eqref{eq:reducing cost} for strictly well-ordering $w$. First, note that equality in~\eqref{eq:reducing cost} holds if and only if equality in~\eqref{eq:goal ineq} holds for every $j$ and $k$. In turn, by inspecting~\eqref{eq:correct ordering} and~\eqref{eq:correct ordering 2} and recalling the definition of strictly well-ordering (see Def.~\ref{def:well-ordering}), we see that equality in~\eqref{eq:goal ineq} holds if and only if $x_j = x_{\sigma(j)}$ or $x_{\ell_k} = x_{\ell_k+1}$ for every $k$ and every $j$. Hence, for equality to hold in~\eqref{eq:reducing cost} we must have either $x_{\ell_k} = x_{\ell_k+1}$ for every $k$, or $x_j = x_{\sigma(j)}$ for every $j$. In the first case we have
\begin{align*}
    \sum_{k\in A} \delta_{x_k} = \sum_{k\in C} \delta_{x_k} + \sum_{k=1}^m \delta_{x_{\ell_k}} = \sum_{k\in C} \delta_{x_k} + \sum_{k=1}^m \delta_{x_{\ell_k+1}} = \sum_{k\in B(A)} \delta_{x_k}.
\end{align*}
In the second case, we have
\begin{align*}
    \sum_{k\in A} \delta_{x_k} = \sum_{k \in C} \delta_{x_k} + \sum_{k=1}^m \delta_{x_{\ell_k}} = \sum_{k\in D} \delta_{x_{\sigma(k)}} + \sum_{k=1}^m \delta_{x_k} = \sum_{k\in D} \delta_{x_k} + \sum_{k=1}^m \delta_{x_{\ell_k}} = \sum_{k\in B(A)^c} \delta_{x_k}.
\end{align*}
This completes the proof.
\end{proof}

We can now complete the proof of Proposition~\ref{prop:well-ordering}. 

\begin{proof}[Proof of Proposition~\ref{prop:well-ordering}] Let $A \subset \{1,...,2n\}$ and $|A| = n$. Let $f_A$ be defined as before. Suppose that $\mathrm{Osc}(f_A) > 1$, or equivalently (see property~\ref{it:oscillation property}), $A_e \neq A\neq A_o$. Then by Lemma~\ref{lem:swap}, we can construct $A_1 \coloneqq B(A)$ such that $\mathrm{Osc}(f_{B(A)}) < \mathrm{Osc}(f_A)$. Moreover, by~\eqref{eq:reducing cost}, the cost associated to the configuration of $A_1$ is smaller or equal than the cost of $A$. If $\mathrm{Osc}(f_{A_1}) = 1$, then by~\ref{it:oscillation property} we must have $A_1 = A_o$ or $B(A) = A_e$. Otherwise we can keep repeating the previous step, i.e., setting $A_k = B(A_{k-1})$, until we obtain $\mathrm{Osc}(f_{A_k}) = 1$, and therefore $A_k = A_e$ or $A_k = A_o$. As the cost does not increase at each iteration, we conclude that
\begin{align*}
    c_n(x_{A_e}) + c_n(x_{A_o}) \leq c_n(x_A) + c_n(x_{A^c}).
\end{align*}
This iterative procedure to reduce the cost is illustrated in Figure~\ref{fig:swap}. As $A$ was arbitrary, we conclude that~\eqref{eq:n well-ordering} holds. The statement about the equality case for strictly well-ordering $w$ follows from the corresponding equality statement in Lemma~\ref{lem:swap}. 
\end{proof}

\begin{figure}[ht!]
    \centering
    \includegraphics[scale=0.37]{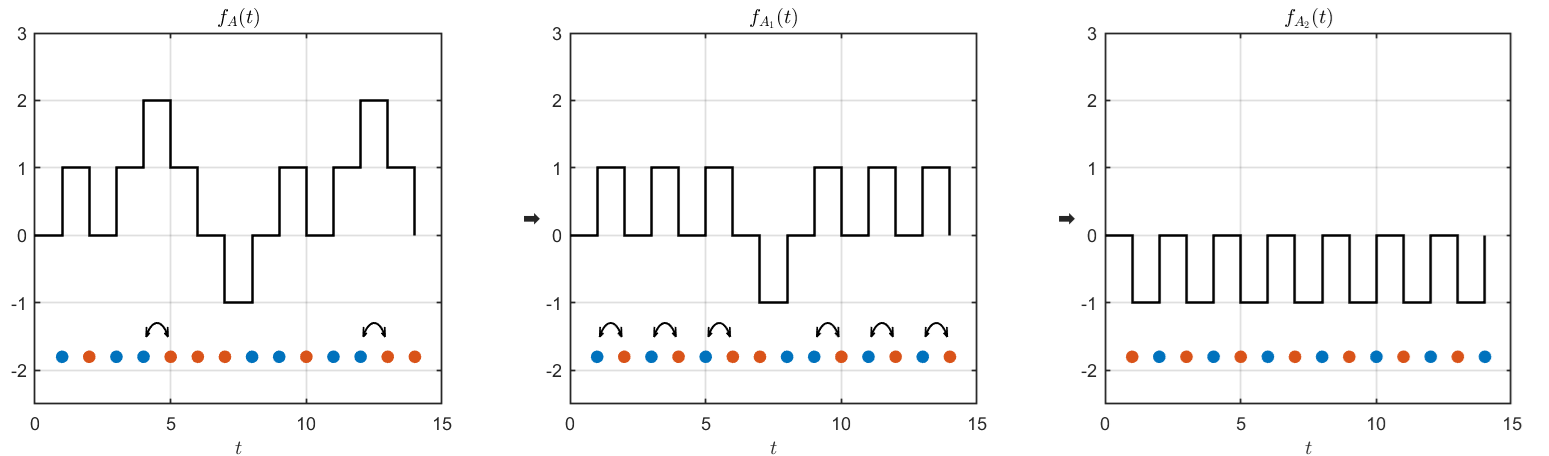}
    \caption{Visual illustration of the swapping procedure to reduce the cost in the case $n=7$ and $A = \{1, 3, 4, 8, 9, 11, 12\}$. Here $A_1 = B(A)$ and $A_2 = B(A_1)$. Moreover, the blue markers represent the elements in $A$, $A_1$ and $A_2$, while the red markers represent the elements in the complementary sets.}
    
 \label{fig:swap}
\end{figure}
\subsection{Sufficient conditions} We now show how the Seidl conjecture follows from Proposition~\ref{prop:well-ordering}. For this, we shall use the following lemma, which is a restatement of classical results in optimal transport. As the assumptions are rather different from previous works \cite{Amb00,Vil03}, we briefly sketch the proof below. 
\begin{lemma}[Two marginal case]\label{lem:strict 2 marginal}
Let $w: I \times I\rightarrow \R \cup{+\infty}$ be well-ordering, and $\rho_1, \rho_2 \in \mathcal{P}(I)$ be non-atomic measures with $\mathrm{supp}(\rho_j) = I_j = [a_j,b_j]$ satisfying $b_1 \leq a_2$. Then there exists a unique non-decreasing map (up to $\rho_1$-null sets) $T:I_1\rightarrow I_2$ such that $T^\# \rho_1 = \rho_2$. Moreover, the plan $\gamma = (\mathrm{id}, T)^\# \rho_1$ is an optimizer of
\begin{align*}
    \min_{\gamma \in \Pi(\rho_1,\rho_2)} \int_{I_2} w(x,y) \mathrm{d} \gamma(x,y).
\end{align*}
In addition, if $w$ is strictly well-ordering, then this plan is the unique optimizer.
\end{lemma}

\begin{proof} The fact that there exists a unique (up to $\rho_1$-null sets) non-decreasing transport map $T$ is rather classical, see e.g. \cite[Theorem 3.1]{Amb00}, \cite[Remarks 2.19]{Vil03}, \cite[Remark 1.23]{San15}. Hence, to show optimality and uniqueness in the strictly well-ordering case, it suffices to show that any optimal plan $\gamma$ is concentrated on the graph of a increasing transport map. The optimality in the (not necessarily strictly) well-ordering case then follows by approximation. 

So let $\gamma$ be an optimizer and let $(x,y), (x',y') \in \mathrm{supp}(\gamma)$ with $x\leq x'$. Since $\mathrm{supp}(\gamma) \subset [a_1,b_1]\times [a_2,b_2]$ and $b_1\leq a_2$, we have $x<y$ and $x'< y'$ (except when $x=x' = b_1$ in the special case $a_2 = b_1$, which is however irrelevant as $\rho_1$, and consequently $\gamma$, is non-atomic). Since $\gamma$ is an optimizer, its support must be $w$-cyclically monotone, which implies that $x\leq x'\leq y \leq y'$ by the definition of strictly well-ordering. In particular, if $x=x'$ we have $y=y'$, and therefore $\mathrm{supp}(\gamma) \subset \{(x,\phi(x)): x\in [a_1,b_1] \mbox{ for some function $\phi$} \}$. Moreover, the previous inequality also shows that $y' =\phi(x') \geq y =\phi(x)$ for $x'\geq x$, which implies that $\phi$ is non-decreasing. As $\pi_1^\# \gamma =\rho_1$, one can then use the desintegration theorem \cite[Theorem 2.28]{AFP00} to prove that $\phi = T$ $\rho_1$-a.e. (see \cite[Proposition 2.1]{Amb00}), which completes the proof.
\end{proof}

\begin{proof}[Proof of first direction in Theorem~\ref{thm:main}] The proof here is the same as in \cite[Theorem 1.1]{CDD15}. For the sake of completeness, we briefly outline the main steps here. 

Let $\gamma$ be a symmetric optimizer of $F_{\rm OT}(\rho)$. As any optimizer is supported away from the coalescence set $\{x\in I_n: x_i = x_j \quad \mbox{for some $i\neq j$}\}$ (see \cite[Corollary 2.6]{CDD15}\footnote{or Lemma~\ref{lem:support} under Assumption~\ref{assump:interaction} on the interaction} for a proof of this statement), it follows that the restriction $\gamma_\ast \coloneqq \gamma \rvert_{O}$ of $\gamma$ to the simplex
\begin{align*}
    O = \{x \in I_n: x_1<x_2<... x_n \in I\}
\end{align*}
has measure $1/n!$ and $\gamma = \sum_{\sigma \in \mathcal{S}_N} \sigma^\# \gamma_\ast$, where $\sigma^\#$ denotes the push-forward via a coordinate permutation $\sigma\in \mathcal{S}_n$. As the support of $\gamma_\ast$ is $c$-monotone (by optimality), we can use Proposition~\ref{prop:well-ordering} to show that 
\begin{align*}
     \quad d_i^- \coloneqq \min \{ x_i : x=(x_1,...,x_n) \in \mathrm{supp}(\gamma_\ast) \}\quad \mbox{and}\quad d_i^+ \coloneqq \max \{ x_i : x= (x_1,...,x_n) \in \mathrm{supp}(\gamma_\ast)\}
\end{align*}
satisfy $d_i^+ \leq d_{i+1}^-$ for $i=1,..., n-1$. Consequently, $\mathrm{supp}(\gamma_\ast) \subset \prod_{j=1}^n [d_j^-, d_j^+]$. One can then use that the sets $[d_i^-,d_i^+]$ are essentially disjoint and $\rho$ is non-atomic to prove that the restriction $\rho_i \coloneqq \rho\rvert_{[d_i^-,d_i^+]}$ satisfies 
\begin{align*}
    \int_{d_i^-}^{d_i^+} \rho = \frac{1}{n}\quad \mbox{and}\quad (n-1)! \pi_i^\#\gamma_\ast = \rho_i \mathrm{d} x,
\end{align*}
where $\pi_i : I_n \rightarrow I$ denotes the projection on the $i^{th}$ coordinate. To complete the proof, the idea now is to look at the reduced pair densities $\gamma_{i,j} = (n-1)!(\pi_i, \pi_j)^\# \gamma_\ast$ and use Lemma~\ref{lem:strict 2 marginal} to prove that the cost of $\gamma$ is at least as large as the cost of the Seidl plan. More precisely, we note that
\begin{align*}
    \frac{1}{n} \int_{I_n} c_n(x) \mathrm{d} \gamma(x) = (n-1)! \int_O c_n(x) \mathrm{d} \gamma_\ast =  \sum_{i\neq j} \int_{I\times I} w(x,y) \mathrm{d} \gamma_{i,j}(x,y).
\end{align*}
On the other hand, similar considerations show that the cost of the Seidl plan in~\eqref{eq:optimizer} is given by
\begin{align*}
    \frac{1}{n} \sum_{i\neq j} \int_{I_2} w(x,y) \mathrm{d} (T^{(i)}, T^{(j)})^\# \rho &=\frac{1}{n} \sum_{i \neq j} \int_{I_2} w(x,y) \mathrm{d} (T^{(i)}, T^{(j)})^\# \left(\sum_{k=1}^n \rho_k\right) \\
    &= 2 \sum_{i< j} \int_{I\times I} w(x,y) \mathrm{d} (\mathrm{id}, T^{(j-i)})^\# \rho_i
\end{align*}
where $T^{(j-i)}$ is the $j-i$ self-composition of the Seidl transport map. Note that $\rho = 0$ a.e. in $[d_i^+,d_{i+1}^-]$; therefore, $\rho_i = \rho\rvert_{[d_i,d_{i+1}]}$ as measures and the Seidl map $T$ satisfies $T^\# \rho_i = \rho_{i+1}$. In particular, $(T^{(j-1)})^\#\rho_i = \rho_j$, and, since $T^{(j-1)}$ is non-decreasing in $[d_i^-,d_i^+]$, Lemma~\ref{lem:strict 2 marginal} implies that $\int c_n \mathrm{d} \gamma_\rho \leq \int c_n \mathrm{d} \gamma$ and therefore the Seidl plan is an optimizer. To prove uniqueness in the strictly well-ordering case, one can use the uniqueness in Lemma~\ref{lem:strict 2 marginal} to show that $\gamma_{i,j} =(\mathrm{Id}, T^{(j-i)})^\# \rho_i$, and therefore, $\gamma_\ast$ is supported on the graph of $(\mathrm{id}, T, T^{(2)},..., T^{(n-1)})$. 
\end{proof}
\subsection{Necessary conditions} We now prove the only if part in Theorem~\ref{thm:main}. For this, it suffices to consider the two marginal case. More precisely, recalling that under the assumption that $F_{\rm OT}(\rho)<\infty$ any optimizer must have $c$-cyclically monotone support, the necessity of the well-ordering condition follows from the following simple proposition.

\begin{proposition} Let $x_1<x_2<x_3<x_4 \in I$ for some interval $I \subset \R$, then there exists $\rho \in \mL^1(I)$ such that $(x_1,x_3)$ and $(x_2,x_4) \in \mathrm{supp}(\gamma_\rho)$, where $\gamma_\rho$ is the Seidl plan~\eqref{eq:optimizer}.
\end{proposition}

\begin{proof} Let $I=[a,b]$, then by construction of $\gamma_\rho$, we have $\mathrm{supp}(\gamma) = \{(x,T(x)): x \in \mathrm{supp}(\rho) \}$, where $T$ is the unique (up to $\rho$ null sets) transport map such that $T^\# \rho = \rho$ and $T$ is monotone on the intervals $[a,d_1]$ and $[d_1,b]$, where $d_1$ satisfies $\rho([a,d_1]) = 1/2$. So let us choose $\rho$ such that 
\begin{align}
    \rho \in C(I;\R_+),\quad \rho(x) >0 \quad \mbox{for any $x\in I$, and}\quad \int_{x_1}^{x_3}\rho(y) \mathrm{d} y = \int_{x_2}^{x_4} \rho(y) \mathrm{d} y = \frac12 . \label{eq:some properties}
\end{align}
Such a $\rho$ is easy to construct, e.g., by gluing piecewise linear functions together.
Since $\rho$ is strictly positive, for any $x\in I$ there exists at most one point $x' >x$ such that $\int_x^{x'} \rho(y) \mathrm{d} y = 1/2$. In particular, it follows from~\eqref{eq:some properties} that the unique point $x_m \in I$ such that $\int_a^{x_m} \rho = 1/2$ satisfies $x_2 \leq x_m \leq x_3$. Hence, for any $x\in [a,x_m)$, $T(x)$ is the unique point such that $\int_x^{T(x)} \rho = 1/2$. Therefore $T(x_1) = x_3$ and $T(x_2) = x_4$, which implies that $(x_1,x_3),(x_2,x_4) \in \mathrm{supp} \gamma_\rho$.
\end{proof}

\section{Strictly correlated electrons in a quantum ring}
\label{sec:SCE limit}

In this section we shall prove Theorem~\ref{thm:SCE limit}. The proof relies on two main steps. In the first step, we show that any optimal plan is supported away from the periodic diagonal set $D_0$ introduced in~\eqref{eq:periodic diagonal}. The second step in the proof of Theorem~\ref{thm:SCE limit} consists on extending the regularization procedure introduced in \cite{Lew18} to the periodic setting.

\subsection{Off-diagonal support} \label{sec:support} We now show that any optimal plan for $F_{\rm OT}(\rho)$ is supported away from the diagonal. For this, we closely follow the strategy from the previous works \cite{BCD18,CDS19}. To keep the paper concise, we refer rather freely to these works throughout our proofs.

Let $w$ be a function satisfying Assumption~\ref{assump:interaction}, we start by introducing the auxiliary functions
\begin{align}
    m(t) \coloneqq \inf \{ w(x,y) : |x-y|_{\mathbb{T}} \leq t \} \quad \mbox{and}\quad M(t) \coloneqq \sup\{w(x,y) : |x-y|_{\mathbb{T}} \geq t \}. \label{eq:m and M functions}
\end{align}
Note that by the continuity of $w$ and compactness of $I\times I$, the inf respectively sup above are actually attained. Moreover, as $w$ blows-up only on the diagonal set $D_0$ (by assumption), the functions $m:(0,\pi] \rightarrow \R_+$ and $M:(0,\pi]\rightarrow \R_+$ are both non-increasing and satisfy the strong repulsion condition
\begin{align}
    \lim_{t \downarrow 0} M(t) \geq \lim_{t\downarrow 0} m(t) = \infty . \label{eq:strong repulsion}
\end{align}
Furthermore, by construction, we have
\begin{align}
    m(|x-y|_{\mathbb{T}}) \leq w(x,y) \leq M(|x-y|_{\mathbb{T}}),\quad \mbox{for $x,y\in [0,2\pi]$.} \label{eq:interaction bounds}
\end{align}
Thus, $m$ and $M$ are precisely the periodic analogues of the functions $m$ and $M$ appearing in \cite{CDS19}. Moreover, let us also define the (periodic) concentration of $\rho$ as
\begin{align}
    \kappa(\rho,r) = \sup_{x\in I} \int_{|x-y|_{\mathbb{T}}<r} \rho(y) \mathrm{d}y, \label{eq:concentration} 
\end{align}
and introduce the following intrinsic thickening of the set of coalescence points
\begin{align}
    D_w^h \coloneqq \{ x =(x_1,...,x_n) \in I_n: w(x_j,x_k) > h \quad \mbox{for some $j\neq k$} \}. \label{eq:intrisic diagonal}
\end{align}
Thus, following the steps in the proof of \cite[Theorem 1.3]{CDS19}, we obtain the following result.

\begin{lemma}[Support away from diagonal] \label{lem:support} Let $\rho \in \mathcal{P}(I)$ and $r>0$. Let $\gamma \in \Pi(\rho)$ be an optimal plan for~\eqref{eq:MMOT}. Then the following holds. If $\kappa(\rho;r)<\frac{1}{n}$, then whenever we have
\begin{align}
    h > 2(n-1) M(\beta/2), \quad m(\beta) > \frac{F_{\rm OT}(\rho)}{1- n \kappa(\rho;r)},\quad \beta/2 \leq r, \label{eq:conditions for support}
\end{align}
it follows that $\gamma(D_w^h) = 0$, where $D^h_w$ is defined in~\eqref{eq:intrisic diagonal}.
\end{lemma}

\begin{proof} The proof is the same as in \cite[Theorem 1.3, pp. 15-16]{CDS19}. The only difference is that we replace the standard euclidean balls $B(x,\gamma)$ used there by their periodic counterparts $B_{\rm per}(x,\gamma) = \{y : |x-y|_\T < \gamma\}$. We also note that the same proof works in arbitrary dimensions.
\end{proof}

\subsection{The Lewin-Bindini-De Pascale regularization in the periodic setting} \label{sec:Lewin construction}
We now adapt the construction introduced by Lewin \cite{Lew18} to the periodic case. For this, let $\chi \in C^\infty_c(\R;\R_+)$ be a radial function with support on the unit ball and such that $\int \chi^2 = 1$. Set 
\begin{align*}
    \chi_{\eta}(x) \coloneqq \eta^{-\frac12} \chi\left(\frac{x}{\eta}\right)\quad \mbox{and}\quad \chi_{\eta, z}(x) \coloneqq \chi_\eta(x-z) \quad (\eta > 0 , z \in \R).
\end{align*}
For $\rho \in \mH^1_{\rm per}(I)$, we let $\widetilde{\rho} \in \mH^1_{\rm loc}(\R)$ denote the periodization of $\rho$, i.e.,
\begin{align*}
    \widetilde{\rho}(x) =\rho(x\, \mathrm{mod} \,2\pi \Z) = \sum_{\ell\in 2\pi \Z} \rho(x+\ell),
\end{align*}
where in the last expression we set $\rho(x) = 0$ for $x\not \in I$. Then, for $\gamma \in \mathcal{P}(I_n)$ we can define $\Gamma_\eta$ as 
\begin{align}
    \Gamma_\eta = \sum_{k,\ell \in (2\pi\Z)^n}  \iint_{I_n \times \R^n} \sqrt{\rho}^{\otimes n} |\chi_{\eta,z_1+k_1}... \chi_{\eta,z_n+k_n}\rangle \langle\chi_{\eta,z_1+\ell_1}.... \chi_{\eta,z_n+\ell_n}|\sqrt{\rho}^{\otimes n}  \prod_{j=1}^n \frac{\chi_\eta(y_k-z_k)^2}{(\widetilde{\rho} \ast \chi_\eta^2)(z_k)} \mathrm{d} \gamma(y) \mathrm{d} z, \label{eq:regularized DM}
\end{align}
where $\sqrt{\rho}^{\otimes n}$ is the operator of multiplication by $\sqrt{\rho}(x_1)\sqrt{\rho}(x_2)... \sqrt{\rho}(x_n)$, $\mathcal{X}_{\eta,z} \coloneqq |\chi_{\eta,z_1}... \chi_{\eta,z_n}\rangle$ is the Slater determinant
\begin{align*}
\mathcal{X}_{\eta,z}(x_1,...,x_n) = \frac{1}{\sqrt{n!}} \mathrm{det} \begin{pmatrix} \chi_{\eta,z_1}(x_1) &... & \chi_{\eta,z_n}(x_n) \\
\vdots &  \ddots & \vdots\\
 \chi_{\eta,z_1}(x_n) & ... & \chi_{\eta,z_n}(x_n)\end{pmatrix},
\end{align*}
and $P = |\mathcal{X}_{\eta,z}\rangle \langle \mathcal{X}_{\eta,w}|$ is the usual rank one operator defined as $P \Psi = \mathcal{X}_{\eta,z} \inner{\mathcal{X}_{\eta,w},\Psi}$. Note that, since the regularized function
\begin{align*}
    \mathbb{Q}_\eta(z) = \int_{I_n} \prod_{j=1}^n \chi_\eta(y_k-z_k)^2 \mathrm{d} \gamma(y) \in C^\infty_c(\R^n;\R_+)  
\end{align*}
has support in an $\eta$ neighborhood of $I_n$, only finitely many terms in the sum over $(2\pi\Z)^n\times (2\pi\Z)^n$ are non-zero. In particular, the operator $\Gamma_\eta$ in~\eqref{eq:regularized DM} is well-defined. 

A direct calculation then yields the following result.
\begin{proposition}[Density of $\Gamma_\eta$] \label{prop:regularized} Suppose that $\mathrm{supp}(\gamma) \subset I_n \setminus D_\alpha \coloneqq \{x\in I_n: |x_i-x_j|_{\mathbb{T}} \geq \alpha, \quad \mbox{for any } i\neq j \}$ for some $\alpha >0$. Then for any $\eta< \alpha/4$, the regularized density matrix $\Gamma_\eta$ satisfies
\begin{align*}
    \rho_{\Gamma_\eta}(x) = n \rho(x), \quad \mbox{for $x\in I$ and } \quad \mathrm{supp}(\Gamma_\eta(x,x)) \cap I_n \subset I_n \setminus D_{\alpha-4\eta}.
\end{align*}
Moreover, if $\sqrt{\rho} \in \mH^1_{\rm per}(I)$, then $\Gamma_\eta$ is periodic and satisfies the following kinetic energy bound
\begin{align}
    \mathrm{Tr} (-\Delta) \Gamma_\eta = n \left(\int_{I} |\nabla \sqrt{\rho}(x)|^2 \mathrm{d} x + \frac{1}{\eta^2} \int_{\R} |\nabla \chi(x)|^2 \mathrm{d} x\right) .  \label{eq:kinetic energy bound}
\end{align}
Furthermore, the $n$-particle density $\Gamma_\eta(x,x)$ satisfies $\Gamma_\eta(x,x) \rightharpoonup \gamma$ in the weak sense of probability measures.
\end{proposition}

\begin{proof}
    Under the assumption that $\mathrm{supp}(\gamma) \subset I_n \setminus D_\alpha$, we have 
    \begin{align*}
        \mathrm{supp}(\mathbb{Q}_\eta) \subset \{z = (z_1,...,z_n) \in \R^n: |z_i-z_j|_{\T} \geq \alpha-2\eta \quad \mbox{for any $i\neq j$} \}.
    \end{align*}
    Since $\mathrm{supp}(\sum_{\ell \in (2\pi\Z)^n} \mathcal{X}_{z+\ell,\eta}) \subset \{x \in \R^n: |x_i|_{\T} < \eta \mbox{ for any $i$} \}$, we have $\mathrm{supp}(\Gamma_\eta(x,x))\cap I_n \subset I_n\setminus D_{\alpha-4\eta}$ as desired. Moreover, note that $\chi_{z_i+k_i}$ and $\chi_{z_j+\ell_j}$ have disjoint support for $z\in \mathrm{supp}(\mathbb{Q}_\eta)$ and $i\neq j$ or $k_i\neq \ell_j$. Consequently, using that 
    \begin{align*}
        \sum_{\ell_j \in (2\pi\Z)} (\chi_\eta^2 \ast \rho)(z_j+\ell_j) = (\chi_\eta^2 \ast \widetilde{\rho})(z_j),
    \end{align*} 
   we can first integrate over $x_2,...,x_n$ and then perform the other integrals in the appropriate order to obtain
    \begin{align*}
        \rho_{\Gamma_\eta}(x_1) &= n \sum_{\ell_1 \in (2\pi\Z)}  \int_{I_n\times \R}\rho(x_1) \chi_\eta(x_1-z_1-\ell_1)^2 \frac{\chi_\eta(y_1-z_1)^2}{(\widetilde{\rho} \ast \chi_\eta^2)(z_1)} \\ 
        & \qquad \qquad \qquad \qquad \qquad \underbrace{\left( \sum_{\ell_2,...\ell_n \in (2\pi\Z)} \int_{\R^{n-1}} \prod_{k=2}^n (\chi_{\eta}^2\ast \rho)(z_k+\ell_k)\prod_{k=2}^n \frac{\chi_\eta(y_k-z_k)^2}{(\widetilde{\rho} \ast \chi_\eta^2)(z_k)}\mathrm{d} z_2...\mathrm{d}z_n\right)}_{=1} \mathrm{d} \gamma(y) \mathrm{d} z_1 \\
        &=  \sum_{\ell_1 \in (2\pi\Z)} \int_{I \times \R} \rho(x_1) \chi_\eta(x_1-z_1-\ell_1)^2 \frac{\chi_\eta(y_1-z_1)^2}{(\widetilde{\rho} \ast \chi_\eta^2)(z_1)} \rho(y_1) \mathrm{d} y_1 \mathrm{d} z_1 \\
        &= \rho(x_1) \int_{\R} \frac{\chi_\eta(x_1-z_1)^2}{(\widetilde{\rho} \ast \chi_\eta^2)(z_1)} \left(\sum_{\ell_1 \in (2\pi\Z)} \int_I \chi_{\eta}(y_1-z_1+\ell_1)^2 \rho(y_1) \mathrm{d} y_1\right) \mathrm{d} z_1 = \rho(x_1).
    \end{align*}

    To see that $\Gamma_\eta$ is periodic, note that
    \begin{align*}
        \sum_{\ell,k \in (2\pi\Z)^n} |\chi_{\eta,z_1+k_1}... \chi_{\eta,z_n+k_n}\rangle \langle \chi_{\eta,z_1+\ell_1},..., \chi_{\eta,z_n+\ell_n}| = |\widetilde{\chi}_{\eta,z_1} .. \widetilde{\chi}_{\eta,z_n} \rangle \langle \widetilde{\chi}_{\eta,z_1}... \widetilde{\chi}_{\eta,z_n} |,
    \end{align*}
    where $\widetilde{\chi}_{\eta,z_j}$ denotes the periodization of $\chi_{\eta,z_j}$. Consequently, if $\sqrt{\rho}$ is periodic, then so are the functions $\widetilde{\chi}_{\eta,z_j} \sqrt{\rho}$, and therefore, also $\Gamma_\eta$.

    The calculation for the kinetic energy is identical to the one in \cite{Lew18}.
\end{proof}

\subsection{The strongly interacting limit} We can now combine Lemma~\ref{lem:support} and Proposition~\ref{prop:regularized} to complete the proof of Theorem~\ref{thm:SCE limit}.

\begin{proof}[Proof of Theorem~\ref{thm:SCE limit}]
Let $\gamma_{\rm opt}$ be an optimal plan for $F_{\rm OT}(\rho)$. By Lemma~\ref{lem:support}, its support is contained outside the set $D^h_w$ defined in~\eqref{eq:intrisic diagonal} for some $h>0$. As $\{w=\infty\} = \{(x,y) : |x-y|_{\mathbb{T}} = 0\}$ by Assumption~\ref{assump:interaction}, we have $\{c = \infty \} = \{ (x_1,...,x_n) \in I_n: |x_i-x_j|_{\mathbb{T}} = 0 \quad \mbox{for some $i\neq j$}\} = D_0 \subset D_w^h$. Thus, $\gamma_{\rm opt}$ is supported away from the periodic coalescence set $D_0$, and there exists $\alpha>0$ such that $\mathrm{supp}(\gamma_{\rm opt}) \cap D_\alpha = 0$. Thus, choosing $\eta< \alpha/4$ and defining $\Gamma_\eta$ via~\eqref{eq:regularized DM}, we obtain a trial density matrix for $F_{\rm per}^\varepsilon(\rho)$. In particular,
\begin{align*}
    \limsup_{\varepsilon\downarrow 0} F_{\rm per}^\varepsilon(\rho) \leq \limsup_{\varepsilon \downarrow 0} \varepsilon T(\Gamma_\eta) + \int_{I_n} c_n(x) \Gamma_\eta(x,x) \mathrm{d} x = \int_{I_n} c_n(x) \Gamma_\eta(x,x) \mathrm{d} x, \quad \mbox{for $\eta < \alpha/4$.}
\end{align*}
Now, since the support of $\Gamma_\eta$ is contained in $I_n\setminus D_{\alpha-4\eta}$ for any $\eta < \alpha/4$ and since $w$ is uniformly bounded on this set (by continuity), the convergence $\Gamma_\eta(x,x) \rightharpoonup \gamma_{\rm opt}$ implies that $\limsup_{\varepsilon \downarrow 0} F_{\rm per}^\varepsilon(\rho) \leq \int c_n \mathrm{d} \gamma_{\rm opt} = F_{\rm OT}(\rho)$. The opposite ($\liminf$-)inequality is trivial because the kinetic energy is always non-negative.
\end{proof}

Let us end this section by remarking that Theorem~\ref{thm:SCE limit} can be extended to periodic systems in arbitrary dimensions $\T^d = \R^d/(2\pi \Z)^d$. Precisely, if we define
\begin{align*}
    |x|_{\mathbb{T}^d} := \min\{|x-k| : k \in (2\pi \Z)^d\}
\end{align*}
and let $D_0(\T^d)$ be the periodic diagonal set on $\T^d \times \T^d$. Then the following holds.

\begin{theorem}[Periodic SCE limit in higher dimensions] \label{thm:SCE limit higher dimensions} Let $w:I_d\times  I_d \rightarrow \R_+\cup \{+\infty\}$ be symmetric, continuous and satisfy $\{w = \infty\} = D_0(\T^d)$. Then, for any $\rho \in \mathcal{P}(I_d)$ such that $\sqrt{\rho} \in \mH^1_{\rm per}(I_d)$ we have
\begin{align*}
    \lim_{\varepsilon \downarrow 0} F_{\rm per}^\varepsilon(\rho) = F_{\rm OT}(\rho).
\end{align*}
Moreover, if $w$ is locally $W^{2,\infty}$ in a neighborhood $U$ of $\mathrm{supp}(\gamma_{\rm opt})$ for some optimal plan $\gamma_{\rm opt}$, then
\begin{align}
    F_{\rm per}^\varepsilon(\rho) = F_{\rm OT}(\rho) + \mathcal{O}(\varepsilon^{\frac12}), \label{eq:next-order error}
\end{align}
with a remainder depending on $\rho$ and $\norm{w}_{\mW^{2,\infty}(U)}$.
\end{theorem}

\begin{proof} Since neither the proof of Lemma~\ref{lem:support} nor Proposition~\ref{prop:regularized} are particular to the one-dimensional case, the proof is exactly the same as the proof of Theorem~\ref{thm:SCE limit}. The remainder estimate~\eqref{eq:next-order error} follows as in the proof of \cite[Theorem 2]{Lew18} by using an estimate analogous to eq. (1.7) there.
\end{proof}

\section{From Kohn-Sham to Kantorovich potentials}
\label{sec:potentials}

The goal of this section is to prove Theorem~\ref{thm:Kantorovich potentials}. 

\subsection{Existence of generalized Kantorovich potentials} \label{sec:existence}

We start by studying a generalized notion of Kantorovich potentials. More precisely, these are distributional subgradients of the optimal transport functional in the space $\mH^1_{\rm per}(I)$. Here, we shall establish their existence; later, we investigate their regularity in the case of a bounded (or truncated) cost function. 

Let us start with the following lemma, which is well-known in the literature \cite{Lie83,San15}. As the proof is rather short, we briefly sketch it below. 

\begin{lemma}[Lower semi-continuity of Lieb functional] \label{lem:lower semi-continuity} Let $F_{\rm per}^\varepsilon: \mH^1_{\rm per}(I)\cap \mathcal{P}(I) \rightarrow \R_+ \cup \{+\infty\}$ be the Lieb functional~\eqref{eq:Lieb functional} with domain given by the set of periodic ($n$-)representable densities 
\begin{align}
    \mathrm{dom}\, F_{\rm per}^\varepsilon = \mathcal{R}_{\rm per} = \{\rho \in \mathcal{P}(I) : \sqrt{\rho} \in \mH^1_{\rm per}(I) \}. \label{eq:representable densities}
\end{align}
For $\varepsilon = 0$, we set $F^0_{\rm per}(\rho) \coloneqq F_{\rm OT}(\rho)$ for $\rho \in \mathcal{R}_{\rm per}$. Suppose that $w$ satisfy Assumption~\ref{assump:interaction}. Then, for any $\varepsilon \geq 0$, $F_{\rm per}^\varepsilon$ is convex and lower semi-continuous in $\mH^1_{\rm per}(I)$.
\end{lemma}

\begin{proof} The convexity is immediate since $F_{\rm per}^\varepsilon(\rho)$ is given by a minimization of a linear functional on a convex space, and $\mathcal{R}_{\rm per}$ is convex (by convexity of the gradient \cite[Theorem 7.8]{LL01}).

For the lower semi-continuity, we shall prove it with respect to the weak topology in $\mathcal{P}(I)$. Suppose $\rho_k \rightharpoonup \rho$ in $\mathcal{P}(I)$. If $\rho_k \not \in \mathcal{R}_{\rm per}$, then the statement is trivial. Otherwise, we let $\Gamma_k$ be a density matrix such that $\rho_{\Gamma_k} = n \rho$ and $F_{\rm per}^\varepsilon(\rho_k) +\delta  \geq \varepsilon T(\Gamma_k) + \int c_n |\Gamma_k|^2$. As $c_n \geq 0$, the sequence $\{\Gamma_k\}$ has uniformly bounded kinetic energy. As $\rho_{\Gamma_k} = n \rho_k \rightarrow n \rho$, it is well-known (see, e.g. \cite{DFM08} or \cite[Lemma 4.4]{Cor25c}) that one can extract a subsequence weakly converging (in the space of trace-class operators) to some $\Gamma\geq 0$ with finite kinetic energy and satisfying $\rho_\Gamma = n \rho$ and $T(\Gamma) \leq \liminf_{k} T(\Gamma_k)$. Consequently, $\Gamma_n(x,x)\rightharpoonup \Gamma(x,x)$ in the sense of measures. As $w$ is lower semi-continuous, it follows that $\int c_n \Gamma(x,x) \leq \liminf_k \int c_n \Gamma_k(x,x)$. Combining this with the fact that $T(\Gamma) \leq \liminf_k T(\Gamma_k)$, we conclude that $F_{\rm per}^\varepsilon(\rho) \leq \liminf F_{\rm per}^\varepsilon(\rho_k) + \delta$. As $\delta$ is arbitrary, we are done.

The proof in the optimal tranport case $\varepsilon = 0$ follows from similar arguments, see, e.g. \cite[pp. 5-7]{San15} or \cite[pp. 93-94]{Fri25}.
\end{proof}

As an immediate consequence of the preceding lemma and a standard result in convex analysis, we can establish the existence of representing potentials for the Lieb functional and generalized Kantorovich potentials for the optimal transport problem. As a side remark, we note that the same argument was used to derive the existence of the representing potential in \cite{SPR+24,Cor25b,SPR+25} under the assumption that $w$ is Laplace bounded\footnote{This assumption was used to properly define $H_n(v_\varepsilon(\rho)/\varepsilon ,w/\varepsilon)$ as a self-adjoint operator with form domain $\mH^1_{\rm per}(I_n) \cap \mathcal{H}_n$. However, we remark that, as long as $w$ is non-negative and Lebesgue measurable, one can still define the self-adjoint operator $H_n(v_\varepsilon(\rho)/\varepsilon,w/\varepsilon)$, but this operator might have a different (smaller) form domain. We shall not go into further details here.}. Here we show that this assumption can be replaced by the positivity and (lower semi-)continuity of $w$, and that the same argument applies to the optimal transport problem.

\begin{lemma}[Existence of generalized Kantorovich potentials] \label{lem:existence of potential} For any $\varepsilon \geq 0$, let $F_{\rm per}^\varepsilon(\rho)$ be as in Lemma~\ref{lem:lower semi-continuity} and let $\mathcal{D}_{\rm per}$ be the set introduced in~\eqref{eq:periodic v-rep densities}. Then, for any $\rho \in \mathcal{D}_{\rm per}$, there exists $v_\varepsilon(\rho)\in \mH^{-1}_{\rm per}(I)$ such that
\begin{align*}
    F_{\rm per}^\varepsilon(\xi) - n\inner{v_{\varepsilon}(\rho),\xi} \geq F_{\rm per}^\varepsilon(\rho) - n\inner{v_{\varepsilon}(\rho),\rho}, \quad \mbox{for any $\xi \in \mathcal{P}(I)\cap \mH^1_{\rm per}(I)$,}
\end{align*}
where $\inner{v,\rho}$ denotes the dual pairing in $\mH^{-1}_{\rm per} \times \mH^{1}_{\rm per}$.
\end{lemma}

\begin{proof} Using the 1D GNS inequality,
\begin{align}
    \norm{f}_{\mL^\infty(I)} \lesssim \norm{f}_{\mL^2(I)}^{\frac12} \norm{f}_{\mH^1(I)}^{\frac12}, \label{eq:GNS}
\end{align}
it is not difficult to show that $\mathcal{D}_{\rm per}$ is the relative interior of the set of representable densities $\mathcal{R}_{\rm per} = \{ \rho \in \mathcal{P}(I): \sqrt{\rho} \in \mH^1_{\rm per}(I) \}$ with respect to the $\mH^1_{\rm per}$ topology, see e.g., \cite[Lemma 4.6]{Cor25c}. Since $F^\varepsilon_{\rm per}$ is lower semi-continuous and convex by Lemma~\ref{lem:lower semi-continuity}, we can apply a standard result from convex analysis (see e.g. \cite[Proposition 5.2]{ET99}) to conclude that the subdifferential of $\partial F^\varepsilon_{\rm per}(\rho) \subset \mH^{-1}_{\rm per}(I)$ is non-empty at any density $\rho \in \mathcal{D}_{\rm per}$. Picking any $v_\varepsilon(\rho)\in \mH^{-1}_{\rm per}(I)$ such that $n v_{\varepsilon}(\rho) \in \partial F^\varepsilon_{\rm per}(\rho)$ completes the proof.
\end{proof}

\begin{remark}[Normalized potentials] \label{rem:normalized potential} We shall say that the potential $v_\varepsilon(\rho)$ is normalized if 
\begin{align*}
    F_{\rm per}^\varepsilon(\rho) = n \inner{v_\varepsilon(\rho),\rho}.
\end{align*}
Notice that this can always be achieved by adding a suitable constant to the potential. In the quantum case ($\varepsilon > 0$), this corresponds to setting the ground-state energy of $H_n(v_\varepsilon(\rho)/\varepsilon,w/\varepsilon)$ to zero.
\end{remark}

\subsection{Regularity of generalized Kantorovich potentials} \label{sec:regularity} We now investigate the regularity of the generalized Kantorovich potentials. For this, we shall use the following lemma.

\begin{lemma} \label{lem:simple} (Subgradient from pointwise inequality) Let $c: I_n \rightarrow \R\cup \{+\infty\}$ be a continuous function and suppose that $v\in \mL^1(I)$ satisfy $c-\oplus^n v \geq 0$ (Lebesgue) almost everywhere in $I_n$, where $(\oplus^n v)(x) \coloneqq \sum_{j=1}^n v(x_j)$. Then for any $\gamma \in \mathcal{P}(I_n)$ whose sum of 1D marginals $\rho_\gamma = \sum_{j=1}^n (\pi_j^\#\gamma)$ is a continuous function in $I$, we have
\begin{align}
    \int_{I_n} c(x)\mathrm{d} \gamma(x) - \int_I \rho_\gamma(x) v(x) \mathrm{d}x \geq 0. \label{eq:subgradient continuous}
\end{align}
In addition, if $v \geq v_0(\rho)$ in $\mH^{-1}_{\rm per}(I)$ for some normalized (in the sense of Remark~\ref{rem:normalized potential}) generalized Kantorovich potential and some $\rho \in \mathcal{D}_{\rm per}$, then $v_0(\rho) = v$ in $\mH^{-1}_{\rm per}(I)$.
\end{lemma}
\begin{proof}
    Let $\gamma\in \mathcal{P}(I_n)$ be such that $\rho_\gamma \in C(I)$. Let $\phi \in C^\infty_c(I;\R_+)$ with $\int \phi =1$ and define
    \begin{align*}
            \gamma_\eta(x) \coloneqq \int_{I_n} (\otimes^n \phi_\eta)(x-y) \mathrm{d} \gamma(y) \in C^\infty_c(\R^n;\R_+),
    \end{align*}
    where $\phi_\eta(x) = \eta^{-1} \phi(x/\eta)$ and $(\otimes^n \phi)(x) = \phi(x_1)\phi(x_2)...\phi(x_n)$. Then note that, since $\rho_\gamma  \in C(I)$, we have 
    \begin{align*}
        \gamma(I_n\setminus (\delta,2\pi-\delta)^n) \leq  \left(\int_0^\delta \rho_\gamma + \int_{2\pi-\delta}^{2\pi} \rho_\gamma \right) \leq 2 \norm{\rho_\gamma}_{\mL^\infty} \delta\quad \mbox{for any $0<\delta<\pi$.}
    \end{align*}
    Hence, $\gamma(I_n\setminus (\delta,2\pi-\delta)^n) = 0$. Consequently, standard approximation arguments show that $\gamma_\eta \in C^\infty_c(\R^n)$ satisfies $\gamma_\eta \rightharpoonup \gamma$ in $\mathcal{P}(I_n)$. As $\rho_{\gamma_\eta} = \rho_\gamma \ast \phi_\eta$ and $\norm{\rho_\gamma}_{\mL^\infty}<\infty$, we also have $\rho_{\gamma_\eta} \overset{\ast}{\rightharpoonup} \rho_\gamma$ in $\mL^\infty(\R)$; hence, inequality~\eqref{eq:subgradient continuous} follows by integrating $c-\oplus^n v$ against $\gamma_\eta$ and passing to the limit $\eta \downarrow 0$.

    Next, by fixing $\rho_\gamma = n \xi$ for some $\xi \in \mH^1_{\rm per}(I) \cap \mathcal{P}(I) \subset C(I)$  in~\eqref{eq:subgradient continuous} and minimizing over $\gamma$ we get
    \begin{align*}
        F_{\rm OT}(\xi) -n \int_I v(x) \xi(x) \mathrm{d} x \geq 0.
    \end{align*}
    Since $v_0(\rho)$ is normalized, the preceding inequality with $\xi = \rho$ implies that
    \begin{align*}
        F_{\rm OT}(\rho) - n\inner{v_0(\rho),\rho} + n\inner{v_0(\rho),\rho} - n\int_I v(x) \rho(x) \mathrm{d} x = n\inner{v_0(\rho)-v,\rho} \geq 0.
    \end{align*}
    On the other hand, since the distribution $v-v_0(\rho)$ is non-negative and $\rho(x) \geq c > 0$ for any $x\in I$, we have 
    \begin{align*}
        0\leq \inner{v-v_0(\rho), \xi} \leq \bignorm{\xi/\rho}_{\mL^\infty} \inner{v-v_0(\rho), \rho} \leq \frac{1}{c} \norm{\xi}_{\mL^\infty} \inner{v-v_0(\rho), \rho} \leq 0, \quad \mbox{for any $0 \leq \xi \in \mH^1_{\rm per}(I)$.}
    \end{align*} 
    As any $\xi \in \mH^1_{\rm per}(I)$ is given by the difference of two non-negative functions, this shows that $v_0(\rho) = v$.
\end{proof}

We can now prove the following regularity result for the generalized Kantorovich potentials of the optimal transport problem with truncated cost.

\begin{lemma}[Regularity of generalized Kantorovich potential] \label{lem:regularity} Suppose that the cost function is bounded by $0 \leq c\leq h$ for some $h>0$. Then, every generalized Kantorovich potential $v_0(\rho) \in \mH^{-1}_{\rm per}(I)$ for some $\rho \in \mathcal{D}_{\rm per}$ satisfies $v_0(\rho) \in \mL^\infty(I)$ with the bound
\begin{align}
\mathrm{ess}\sup v_0(\rho) - \mathrm{ess}\inf v_0(\rho) \leq h .\label{eq:potential bound}
\end{align}
\end{lemma}

\begin{proof} \textbf{Step 1: (From distributions to measures)} Without loss of generality, we assume $v_0(\rho)$ is normalized, i.e., $F_{\rm OT}(\rho) - n \inner{v_0(\rho),\rho} = 0$. Thus, from the definition of the generalized Kantorovich potential we have
\begin{align}
    -n \inner{v_0(\rho),\xi} \geq -F_{\rm OT}(\xi) \geq -h, \quad \mbox{for any $\xi \in \mH^1_{\rm per}(I)\cap \mathcal{P}(I)$.}   \label{eq:subgradient inequality}
\end{align}
Hence, 
\begin{align*}
    \inner{h/n-v_0(\rho), \xi} \geq 0 \quad \mbox{for any $\xi \geq 0\in \mH^1_{\rm per}(I)$.}
\end{align*}
Consequently, $h/n -v_0(\rho)$ is a positive distribution in $\mathcal{D}'(\T)$. Therefore, by the Riesz representation theorem in $C(\T)$, there exists a unique non-negative Radon measure $\mu \in \mathcal{M}(\T)$ such that
\begin{align*}
    \inner{h/n-v_0(\rho), f} = \int_{\T} f(x) \mathrm{d} \mu(x), \quad \mbox{for any $f\in \mH^1_{\rm per}(I) \cong \mH^1(\T)$,}
\end{align*}
where the identification $\mH^1_{\rm per}(I) \cong \mH^1(\T)$ is done via the quotient map $\pi : I = [0,2\pi] \mapsto \T = [0,2\pi]/\sim$, $\pi(x) = x \mod 2\pi$. We can therefore write 
\begin{align*}
    v_0(\rho) = v_0(x) \mathrm{d} x - \mu_{s},
\end{align*}
where $h/n-v_0(x) \in \mL^1(\T)\cong \mL^1(I)$ is the Radon-Nikodym derivative of $\mu$ with respect to the Lebesgue (Haar) measure on the torus and $\mu_s\geq 0$ is the singular part. 

\textbf{Step 2: (Regular subgradient)} We now claim that 
\begin{align}
    c(x)-(\oplus^n v_0)(x) \geq 0 \quad \mbox{for almost every $x\in I_n$.} \label{eq:positive claim}
\end{align}
To prove this, first notice that, since $\mu_s \perp \mathrm{d}x $, there exists a Borel set $A\subset \T$  such that $|A| = 0$ and $\mu(\T\setminus A) = 0$. Let $\widetilde{A} = \pi^{-1}(A)$, where $\pi:I \rightarrow \T$ is the quotient map introduced before, then the set $(I\setminus \widetilde{A})^n$ has full Lebesgue measure in $I_n$. In particular, to prove~\eqref{eq:positive claim}, it suffices to show that 
\begin{align*}
    \int_K c(x)-(\oplus^n v_0)(x) \mathrm{d} x \geq 0 \quad \mbox{for any compact set $K \subset (I\setminus \widetilde{A})^n$.}
\end{align*}
For this, let $K$ be such a set and let $\widetilde{K} = \cup_{\ell \in (2\pi \Z)^n} (K+\ell)$ be the periodization of $K$. Then pick a sequence of periodic Lipschitz functions $f_h$ such that $0\leq f_h \leq 1$ and $f_h(x) \downarrow \mathbb{1}_{\widetilde{K}}(x)$; for instance, $f_k(x) \coloneqq 1- \min\{1,k \mathrm{dist}(x,\widetilde{K})\}$. Since $\widetilde{K} \cap I_n \subset (I\setminus \widetilde{A})^n$, the single-particle density of $f_k$ satisfy
\begin{align*}
    \lim_{k\rightarrow \infty} \rho_{f_k}(x) = \lim_{k\rightarrow \infty} \int_{I_{n-1}} \left(f_k(x,x_2,...,x_n) ... + f_k(x_2,...,x_n,x) \right)\mathrm{d} x_2... \mathrm{d} x_n \rightarrow 0 \quad \mbox{for any $x\in \widetilde{A}$.}
\end{align*}
Therefore, as $|\widetilde{K}\cap I_n \setminus K| = 0$, we can apply the dominated convergence theorem and the subgradient inequality~\eqref{eq:subgradient inequality} to obtain
\begin{align*}
    \int_{K} (c-\oplus^n v_0)(x) \mathrm{d} x = \lim_{k\rightarrow \infty} \left(\int_{I_n} (c-\oplus^n v_0)(x)|f_k(x)|^2 \mathrm{d} x + \int_\T \rho_{f_k}(x) \mathrm{d} \mu_s(x) \right)\geq 0,
\end{align*}
which establishes our claim. Moreover, from this claim and the second statement in Lemma~\ref{lem:simple}, we conclude that $\mu_s = 0$. 

\textbf{Step 3: (Uniform bounds)} To complete the proof, we now show that $-h \frac{n-1}{n} \leq v_0 \leq h/n$ almost everywhere in $I$. To this end, first notice that the a.e. bound~\eqref{eq:positive claim} yields $h - \otimes^n v_0 \geq c- \otimes^n v_0 \geq 0$ a.e. in $I_n$, which implies that $v_0 \leq h/n$ a.e. in $I$.  On the other hand, if we define $u = \max \{v_0, - \frac{n-1}{n} h\}$, then for any point $x= (x_1,...,x_n) \in I_n$ such that $c(x) - (\oplus^n u)(x) < 0$, we have 
\begin{align*}
    0>c(x) - \sum_{j\neq i}^n u(x_j) - u(x_i) \geq - \frac{n-1}{n} h - u(x_i) \quad \Rightarrow \quad u(x_i) > -\frac{n-1}{n} h \quad \Rightarrow u(x_i) = v_0(x_i). 
\end{align*}
Hence, $\{x: (c- \oplus^n u)(x) < 0\} \subset \{x:(c - \oplus^n v_0)(x) < 0\}$. Consequently, $c(x)-\oplus^n u(x) \geq 0$ a.e. in $I_n$. We can now use Lemma~\ref{lem:simple} and argue as before to conclude that $0 \geq \int_I (u-v_0)(x) \rho(x) \mathrm{d} x$ and therefore $v_0 = u$ a.e. in $I$, which completes the proof.
\end{proof}

We are now ready to prove the main result of this section, namely, that a generalized Kantorovich potential is in fact a standard (regular) Kantorovich potential for bounded cost functions.

\begin{lemma}[From generalized to classical Kantorovich potentials] Suppose the cost function is continuous and bounded. Then, for any $\rho \in \mathcal{D}_{\rm per}$, any generalized Kantorovich potential $v_0 = v_0(\rho) \in \mL^\infty(I)$ has a continuous representative $v_0 \in C(I)$ satisfying
\begin{align*}
    v_0(x) = \inf_{y \in I_{n-1}} c(x,y_1,...,y_{n-1}) - \sum_{j=1}^{n-1} v_0(y_j).
\end{align*}
In particular $v_0$ is a classical Kantorovich potential.
\end{lemma}

\begin{proof} First, let us introduce the following weak version of the multimarginal $c$-transform: for $u\in \mL^\infty(I)$,
\begin{align*}
    u_c(x) \coloneqq \mathrm{ess} \!\!\!\inf_{y \in I_{n-1}} c(x,y_1,...,y_{n-1}) - \sum_{j=1}^{n-1} u(y_j)
\end{align*}
Then, we claim that, for any $u\in \mL^\infty(I)$ we have
\begin{align}
    u(x) \leq u_c(x) \quad \mbox{for a.e. $x\in I$ if and only if} \quad c(x)- (\oplus^n u)(x) \geq 0 \quad \mbox{for a.e. $x\in I_n$.} \label{eq:obs 1}
\end{align}
Indeed, let $\mathcal{N} \coloneqq \{x \in I_n : c(x) - \oplus^n u(x) < 0 \}$ and $\mathcal{N}_{x} \coloneqq \{ y \in I_{n-1} : c(x,y_2,...,y_n) - \oplus_{j=1}^{n-1} u(y_j) < u(x) \}$, then
\begin{align*}
    |\mathcal{N}| = \int_I \int_{\mathcal{N}_x} \mathrm{d} y \mathrm{d} x = \int_I |\mathcal{N}_x| \mathrm{d} y.
\end{align*}
Hence $|\mathcal{N}| = 0$ if and only if $|\mathcal{N}_x| = 0$ for a.e. $x\in I$. The first statement is equivalent to $c-\oplus^n u \geq 0$ a.e. in $I_n$, while the second is equivalent to $u\leq u_c$ a.e. in $I$. Thus the claim holds. 

Next, note that a similar argument shows that
\begin{align*}
    c(x) - \oplus_{j\neq i}^n u(x_j) - u_c(x_i) \geq 0 \quad \mbox{for a.e. $x \in I_n$ and any $1\leq i \leq n$.}
\end{align*}
Hence, if we set $\bar{u}(x) = \frac{n-1}{n} u(x) + \frac{1}{n} u_c(x)$, it follows that
\begin{align}
    c(x) - (\oplus^n \bar{u})(x) = \frac{1}{n} \sum_{i=1}^n c(x) - \oplus_{j\neq i} u(x_j) - u_c(x_i) \geq 0 \quad \mbox{a.e. in $I_n$.} \label{eq:obs 2}
\end{align}

Now let $v_0 \in \mL^\infty(I)$ be the generalized Kantorovich potential from Lemma~\ref{lem:regularity}. Then, since $c- \oplus^n v_0 \geq 0$ a.e. (see the proof of Lemma~\ref{lem:regularity}), by~\eqref{eq:obs 1} and~\eqref{eq:obs 2}, the function
\begin{align*}
    \bar{v}(x) = \frac{n-1}{n} v_0(x) + \frac{1}{n} (v_0)_c(x)
\end{align*}
satisfies $\bar{v} \geq v_0$ a.e. and $c-\oplus^n \bar{v} \geq 0$ a.e. in $I_n$. By Lemma~\ref{lem:simple}, this implies that $\bar{v} = v_0$ and therefore $v_0 = (v_0)_c$ a.e. in $I$. To complete the proof, we now note that the $c$ transform is regularizing, i.e., $(v_0)_c$ is continuous for any $v_0 \in \mL^\infty(I)$. Indeed, since $c$ is continuous in the compact set $I_n$, it is uniformly continuous. Hence for any $\epsilon>0$, there exists $\delta>0$ such that $|c(x,y_1,...,y_{n-1})-c(x',y_1,...,y_{n-1})| \leq \epsilon$ provided that $|x-x'|\leq \delta$, and therefore,
\begin{align*}
    (v_0)_c(x)\leq \mathrm{ess} \inf_{y\in I_{n-1}} \{c(x',y) - (\otimes_{j=2}^n v_0)(y) + \epsilon\} = (v_0)_c(x') + \epsilon \quad \mbox{for $|x'-x|<\delta$.}
\end{align*}
As we can exchange the roles of $x$ and $x'$, this shows that $(v_0)_c$ is continuous. Since $v_0 = (v_0)_c$ a.e., $(v_0)_c$ is the continuous representative of $v_0$. Moreover, by continuity of $c$ and $(v_0)_c$, we can replace $\mathrm{ess} \inf$ by $\min$ in the definition of $(v_0)_c$, i.e., we have
\begin{align*}
    (v_0)_c(x) = \mathrm{ess} \!\!\!\inf_{y\in I_{n-1}} c(x,y) - (\oplus_{j=2}^n v_0)(y) = \mathrm{ess}\!\!\! \inf_{y\in I_{n-1}} c(x,y) - \oplus^n_{2} (v_0)_c(y) = \min_{y\in I_{n-1}} c(x,y) - (\oplus_{2}^n (v_0)_c(y),
\end{align*}
which completes the proof.
\end{proof}

\begin{remark}[Generalized Kantorovich potentials for vanishing densities] If $\rho \in \mathcal{R}_{\rm per}\setminus \mathcal{D}_{\rm per}$, i.e., $\rho(x) = 0$ for some $x\in I$, then for any Kantorovich potential $v_0 \in \partial F_{\rm OT}(\rho)$, the distribution $v_0 + \alpha \delta_x$ for $\alpha >0$ is also a Kantorovich potential. In particular, Lemma~\ref{lem:regularity} fails in this case.
\end{remark}
\subsection{Local equivalence to truncated costs}
\label{sec:reduction to truncated cost}

We have now seen that generalized Kantorovich potentials are regular potentials for bounded and continuous cost functions. To pass to the case of unbounded costs, we now show that, locally in $\mathcal{D}_{\rm per}$, the optimal transport problem with unbounded cost can be reduced to the problem with a truncated cost. For this, we first prove the following lemma, which is an extension of \cite[Theorem 4.1]{CDS19} (or rather \cite[Remark 4.2]{CDS19}, see also \cite[Proposition 2.5]{BCD18}) to the periodic and quantum cases.

\begin{lemma}[Crude local bound on optimal cost]\label{lem:uniform bound} Let $\rho \in \mathcal{P}(I)$ be such that $\kappa(\rho;r) < \frac{1}{n}$ for some $r>0$ and let $M$ denote the function introduced in~\eqref{eq:m and M functions}. Then we have
\begin{align}
    F_{\rm OT}(\rho) \leq n(n-1) M(r). \label{eq:local bound}
\end{align}
Moreover, if $\sqrt{\rho} \in \mH^1_{\rm per}(I)$, then
\begin{align}
    F^\varepsilon_{\rm per}(\rho) \leq \varepsilon n \left(\int_I |\partial_x \sqrt{\rho}|^2 + \frac{1}{\eta^2} \int_I |\nabla \chi|^2 \right) + n(n-1) M(r-4\eta),\label{eq:local bound quantum case}
\end{align}
for any $0<\eta<r/4$ and any mollifier $\chi\in C_c^\infty(\R)$ as in Proposition~\ref{prop:regularized}.
\end{lemma}

\begin{proof} Define $D_r^2 = \{(x,y) \in I\times I: |x-y|_{\mathbb{T}} < r\}$. Since $\kappa(\rho;r) < 1/n$, we have $\rho(D^2_r(x)) < \frac{1}{n}$ for any $x\in I$, where $D^2_r(x) = \{ y : (x,y) \in D^2_r\} = \{ y : |y-x|_{\T} < r\}$. Therefore, we can apply \cite[Theorem 4.3]{CDS19} to find a plan $\gamma \in \Pi(\rho)$ such that $\gamma(D_r) = 0$ where $D_r\coloneqq \{ (x_1,..,x_n) \in I_n: |x_j-x_k|_\T < r \quad \mbox{for some $j\neq k$}\}$. Hence, from~\eqref{eq:interaction bounds} we find that $F_{\rm OT}(\rho) \leq \int c_n \mathrm{d} \gamma  = \int_{I_n \setminus D_r} c_n \mathrm{d} \gamma \leq n(n-1) M(r)$, which completes the proof of~\eqref{eq:local bound}.

For the estimate on the quantum case ($\varepsilon>0$), we define $\gamma$ as before and use the regularized density matrix $\Gamma_\eta$ for $\eta < r/4$ defined via~\eqref{eq:regularized DM} as a trial state. The estimate then follows from the kinetic energy bound~\eqref{eq:kinetic energy bound}, the fact that $\mathrm{supp}(\Gamma_\eta(x,x)) \cap I_n \subset I_n\setminus D_{r-4\eta}$, and the definition of $M$ (see~\eqref{eq:m and M functions}. 
\end{proof}

One can now combine Lemmas~\ref{lem:support} and~\ref{lem:uniform bound} to obtain the following analogue of \cite[Lemma 5.1]{CDS19}.

\begin{lemma}[Equivalence to truncated costs] \label{lem:truncated cost} Suppose that $w$ satisfies Assumption~\ref{assump:interaction} and $\rho \in \mathcal{P}(I)$ satisfies $\kappa(\rho;r) < 1/n$ for some $r> 0$. Let $\beta$ and $h>0$ be such that
\begin{align}
    0<\beta/2 \leq r, \quad m(\beta) > \frac{n(n-1) M(r)}{1-n \kappa(\rho;r)},\quad \mbox{and}\quad h > 2(n-1) M(\beta/2). \label{eq:support conditions 2}
\end{align}
Then 
\begin{enumerate}[label=(\roman*)]
\item \label{it:truncated optimizers} $\gamma \in \Pi(\rho)$ is a minimizer of $F_{\rm OT}(\rho)$ if and only if it is a minimizer of the optimal transport problem
\begin{align*}
    F_{\rm OT}^{h}(\rho) \coloneqq\inf_{\gamma \in \Pi(\rho)} \int_{I_n} \sum_{j\neq k} w^{h}(x_j,x_k)  \mathrm{d} \gamma(x),\quad \mbox{where $w^{h}(x,y) \coloneqq \min\{w(x,y), h\}$.}
\end{align*}
Moreover, $F_{\rm OT}^h(\rho) = F_{\rm OT}(\rho)$. 
\item \label{it:truncated potentials} If $v\in C(I)$ is a (regular) Kantorovich potential of $F_{\rm OT}^{h}(\rho)$, then it is also a (regular) Kantorovich potential for $F_{\rm OT}(\rho)$.
\end{enumerate}
\end{lemma}

\begin{proof} Let $M^{h}$ and $m^{h}$ be the functions defined as in~\eqref{eq:m and M functions} but for the truncated interaction $w^{h}(x,y) = \min\{w(x,y),h\}$. Then note that, if $\beta$ and $h$ satisfy~\eqref{eq:support conditions 2}, then \begin{align*}
    m^{h}(\beta) = \min\{ m(\beta),h\} \geq \min\{m(\beta), 2(n-1) M(\beta/2)\} = m(\beta) > \frac{n(n-1) M(r)}{1-n\kappa(\rho;r)} 
\end{align*}
and
\begin{align*}
    2(n-1) M(\beta/2) \geq 2(n-1) M^{h}(\beta/2),\quad \mbox{for any $h>0$.}
\end{align*}
Thus, using the upper bound in Lemma~\ref{lem:uniform bound}, we see that $\beta$ and any $h>h'> 2(n-1) M(\beta/2)$ satisfy~\eqref{eq:conditions for support} for the truncated interaction $w^{h}$. Since
\begin{align*}
    D^{h'}_{w^{h}} = \{x= (x_1,...,x_n) \in I_n : w^{h}(x_j,x_k) > h' \quad \mbox{for some $j\neq k$} \} = D^{h'}_w, 
\end{align*}
we conclude from Lemma~\ref{lem:support} that any optimizer $\gamma^h$ of $F^h_{\rm OT}(\rho)$ is supported outside the set $D^{h'}_w$. As $c_n(x) = \sum_{j\neq k} w(x_j,x_k) = \sum_{j\neq k} w^{h}(x_j,x_k) =: c_n^{h}(x)$ for any $x\in I_n\setminus D^{h'}_w$, we obtain
\begin{align*}
    F_{\rm OT}(\rho) \geq F_{\rm OT}^h(\rho) = \int_{I_n} c_n^h(x) \mathrm{d} \gamma^h(x) = \int_{I_n\setminus D^{h'}_w} c_n(x) \mathrm{d} \gamma^h(x) = \int_{I_n} c_n(x) \mathrm{d} \gamma^h(x) \geq F_{\rm OT}(\rho),
\end{align*}
which shows that $F_{\rm OT}(\rho) = F_{\rm OT}^h(\rho)$ and $\gamma^{h}$ is an optimizer of $F_{\rm OT}(\rho)$. Moreover, if we now use that $F_{\rm OT}(\rho) = F_{\rm OT}^h(\rho)$ and $c_n^h \leq c_n$, it is easy to show that any optimizer of $F_{\rm OT}(\rho)$ is also an optimizer for the truncated problem. This completes the proof of~\ref{it:truncated optimizers}.

To prove~\ref{it:truncated potentials} we note that, up to a constant (or normalizing $v$ in the sense of Remark~\ref{rem:normalized potential}), we have $F_{\rm OT}^h(\rho) = n \int v \rho$ and $c_n^h(x) - (\oplus^n v)(x) \geq 0$. So $c_n(x) - (\oplus^n v)(x) \geq 0$ and $n \int v \rho = F_{\rm OT}(\rho)$ thus completing the proof.
\end{proof}

\subsection{Leading asymptotics of the adiabatic potential}

We are now ready to prove Theorem~\ref{thm:Kantorovich potentials}.

\begin{proof}[Proof of Theorem~\ref{thm:Kantorovich potentials}]
Let $\rho \in \mathcal{D}_{\rm per}$ and $v_\varepsilon(\rho) \in \partial F^\varepsilon_{\rm per}(\rho)$ be a sequence of representing potentials, which exists by Lemma~\ref{lem:existence of potential}. Let $r>0$ be such that $\kappa(\rho;r)<1/n$. Since the $\mH^1$ norm controls the  $\mL^1$ norm, hence also the concentration $\kappa(\rho;r)$,  we can find a small $\delta>0$ such that 
\begin{align}
    \kappa_\delta \coloneqq \sup \left\{ \kappa(\xi;r) : \xi \in B_\delta^{\mH^1}(\rho)\cap \mathcal{P}(I)\right\} < 1/n. \label{eq:uniform concentration}
\end{align}
Moreover, since $\rho(x)\geq c >0$ for any $x\in I$, we can use the GNS inequality~\eqref{eq:GNS} to show that, if $\delta>0$ is taken small enough, we also have
\begin{align}
    \sup \left\{ \int_{I} |\partial_x \sqrt{\xi}|^2 : \xi \in B_\delta^{\mH^1}(\rho)\cap \mathcal{P}(I)\right\} < \infty. \label{eq:uniform kinetic energy}
\end{align}
Therefore, by Lemma~\ref{lem:uniform bound}, we have $F^\varepsilon_{\rm per}(\xi) \leq C$ for any $\xi \in B_\delta^{\mH^1}(\rho)$ and $\varepsilon\leq 1$ with a constant $C = C(\rho;\delta)>0$ independent of $\xi$ and $\varepsilon$. Hence, from the subgradient inequality we obtain
\begin{align*}
   n \delta \norm{v_\varepsilon(\rho)}_{\mH^{-1}/\{1\}} = \mathrm{sup}_{\xi \in B_\delta(\rho)\cap \mathcal{P}(I)} n\inner{v_\varepsilon(\rho),\xi-\rho} \leq \mathrm{sup}_{\xi \in B_\delta(\rho)\cap \mathcal{P}(I)}  F^\varepsilon(\xi)-F^\varepsilon(\rho) \leq C(\rho,\delta),
\end{align*}
where $\norm{\cdot}_{\mH^{-1}/\{1\}}$ denotes the canonical norm on the quotient space obtained by identifying potentials in $\mH^{-1}_{\rm per}$ modulo constants. In particular, after normalizing the potentials (in the sense of Remark~\ref{rem:normalized potential}), the set $\{v_\varepsilon(\rho)\}$ is uniformly bounded in $\mH^{-1}_{\rm per}$. Hence, up to subsequences, we can extract a limit $v_\varepsilon(\rho) \rightharpoonup v_0$ in $\mH^{-1}_{\rm per}$. To complete the proof, it suffices to show that $v_0$ is a regular Kantorovich potential for $F_{\rm OT}(\rho)$.

To this end, notice that, since we have a uniform bound on the concentration of any density in $\xi \in B_\delta^{\mH^1}(\rho)$ (namely~\eqref{eq:uniform concentration}), by Lemmas~\ref{lem:truncated cost}~\ref{it:truncated optimizers} and Theorem~\ref{thm:SCE limit}, we can find a $h>0$ independent of $\xi$ such that 
\begin{align*}
    \lim_{\varepsilon \downarrow 0} F_{\rm per}^\varepsilon(\xi) = F_{\rm OT}(\xi) = F_{\rm OT}^h(\xi) \quad \mbox{for any $\xi \in B_\delta^{\mH^1}(\rho)\cap \mathcal{P}(I)$.}
\end{align*}
Consequently, we can use the subgradient inequality for $F^\varepsilon$ and pass to the limit to obtain
\begin{align*}
    \lim_{\varepsilon \downarrow 0} F_{\rm per}^\varepsilon(\xi) - F^\varepsilon_{\rm per}(\rho) - n \inner{v_\varepsilon(\rho),\xi-\rho} = F_{\rm OT}^h(\xi) - F_{\rm OT}^h(\rho) - n \inner{v_0,\xi - \rho} \geq 0, \quad \mbox{for $\xi \in B^{\mH^1}_\delta(\rho)\cap \mathcal{P}(I)$.}
\end{align*}
Using the convexity of $F_{\rm OT}^h$, we then find that the subgradient inequality holds for any $\xi \in \mathcal{R}_{\rm per}$. Hence $v_0$ is a generalized Kantorovich potential of $F_{\rm OT}^h$. By Lemma~\ref{lem:regularity}, we must have $v_0 \in C(I)$; therefore, by Lemma~\ref{lem:truncated cost}~\ref{it:truncated potentials}, $v_0$ is a regular Kantorovich potential of $F_{\rm OT}(\rho)$, thereby completing the proof. 

The uniqueness of the Kantorovich potential in the Lipschitz case follows from classical multimarginal OT theory (see, e.g., \cite[Corollary 2.1]{Pas11}).
\end{proof}
%%%%%%%%%%%%%%%%%%%%%%%%%%%%%%%%%%%%%%%%%%%%%%%%%%
\addtocontents{toc}{\protect\setcounter{tocdepth}{1}}
%%%%%%%%%%%%%%%%%%%%%%%%%%%%%%%%%%%%%%%%%%%%%%%%%%

%\appendix
%\input{auxiliary.tex}
%%%%%%%%%%%%%%%%%%%%%%%%%%%%%%%%%%%%%%%%%%%%%%%%%%
\addtocontents{toc}{\protect\setcounter{tocdepth}{-1}}
%%%%%%%%%%%%%%%%%%%%%%%%%%%%%%%%%%%%%%%%%%%%%%%%%%
\section*{Acknowledgements}

The author is grateful to Gero Friesecke and Mathieu Lewin for fruitful discussions that motivated the topic of this paper.

T.C.~Corso acknowledges funding by the \emph{Deutsche Forschungsgemeinschaft} (DFG, German Research Foundation) - Project number 442047500 through the Collaborative Research Center "Sparsity and Singular Structures" (SFB 1481). 

%%%%%%%%%%%%%%%%%%%%%%%%%%%%%%%%%%%%%%%%%%%%%%%%%%
\addtocontents{toc}{\protect\setcounter{tocdepth}{2}}
%%%%%%%%%%%%%%%%%%%%%%%%%%%%%%%%%%%%%%%%%%%%%%%%%%
%\appendix
%%%%%%%%%%%%%%%%%%%%%%%%%%%%%%%%%%%%%%%%%%%%%%%%%%

%%%%%%%%%%%%%%%%%%%%%%%%%%%%%%%%%%%%%%%%%%%%%%%%%%

\section*{Data availability}
No datasets were generated or analysed during the current study.

\section*{Competing interests}

The author has no competing interests to declare that are relevant to the content of this article.

\appendix
\section{Well-ordering interactions on a convex graph}

In this appendix, we show that interactions of the form~\eqref{eq:graph example} are well-ordering. For this, let $0\leq x_1\leq x_2 \leq x_3 \leq x_4$. Then first, from the convexity of $f$, it is not hard to show that the convex quadrilateral with vertices at $v_j = (x_j,f(x_j))$ with $j=1,2,3,4$ has inner diagonals given by line segments $[v_1,v_3]$ and $[v_2,v_4]$, see Figure~\ref{fig:convex graph}. In particular, from the triangle inequality (for the Euclidean norm in $\R^2$) we have
\begin{align}
    |v_3-v_1| + |v_4-v_2| \geq \min\{ |v_2-v_1|+ |v_4-v_3|,|v_4-v_1|+ |v_3-v_2|\}. \label{eq:diagonal inequality}
\end{align}
\begin{figure}[ht!]
    \centering
    \includegraphics[scale=0.366]{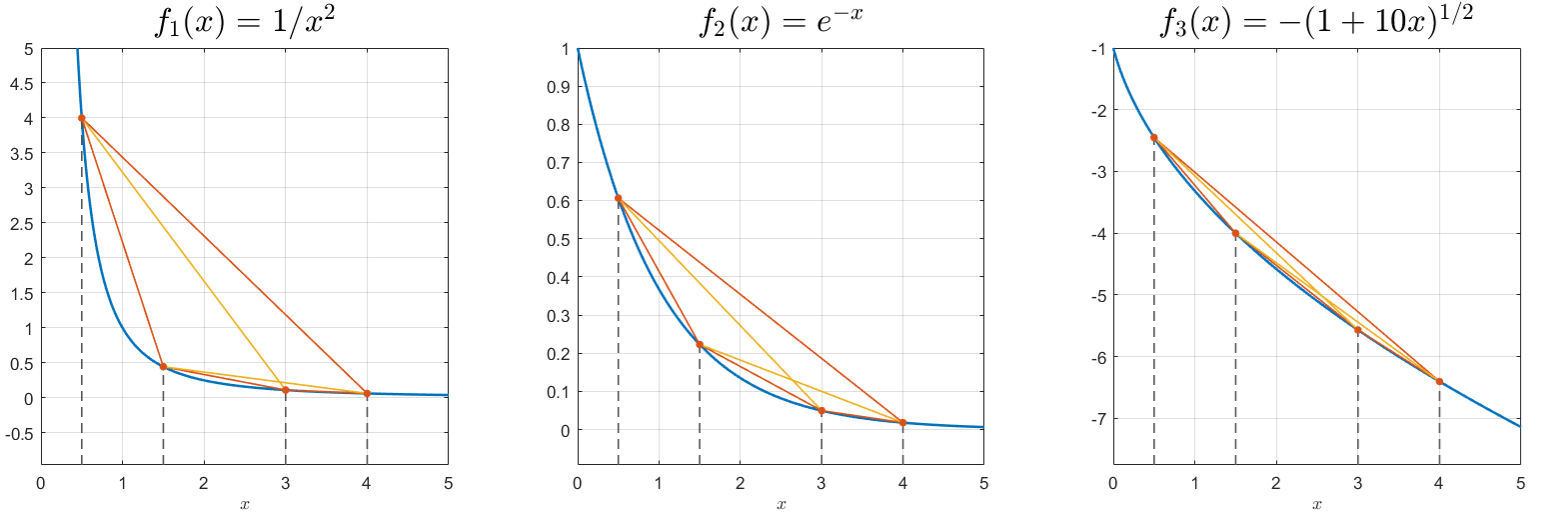}
    \caption{Illustration of the convex quadrilaterals with vertices $(x_j,f(x_j))$ for $\{x_1,x_2,x_3,x_4\} = \{0.5, 1.5, 3, 4\}$ (in red) with diagonals (in yellow) on the graph of different convex and non-increasing functions $f$ (in blue).}
    
 \label{fig:convex graph}
\end{figure}
On the other hand, as $f$ is non-increasing we also have
\begin{align*}
    |f(x_1)-f(x_3)| \geq \max\left\{|f(x_2)-f(x_1)|, |f(x_2) - f(x_3)| \right\},
\end{align*}
which implies that $|v_3-v_1| \geq \max \{|v_3-v_2|,|v_2-v_1|\}$. Similar arguments show that $|v_4-v_2| \geq \max\{|v_3-v_2|,|v_4-v_3|\}$ and $|v_4-v_1| \geq \max\{|v_1-v_3|,|v_4-v_2|\} \geq |v_3-v_2|$. Together, these inequalities imply that
\begin{align}
    &\max\{d_{13},d_{24}\} \geq \max\{d_{12},d_{34}\}, \quad \min\{d_{13}, d_{24}\} \geq \min \{ d_{12}, d_{34}\}, \quad\mbox{and}\label{eq:min and max inequality} \\
    &d_{14} \geq \max\{d_{13},d_{24}\} \geq \min\{d_{13},d_{24}\} \geq d_{23}, \label{eq:min inequality}
\end{align}
where $d_{ij} := |v_i-v_j|=\sqrt{(x_i-x_j)^2 + \left(f(x_i)-f(x_j)\right)^2}$.

As $g$ is non-increasing, the two inequalities in~\eqref{eq:min and max inequality} imply that
\begin{align*}
    w(x_1,x_3)+w(x_2,x_4) = g(d_{13}) + g(d_{24}) \leq g(d_{12}) + g(d_{34}) = w(x_1,x_2) + w(x_3,x_4).
\end{align*}
Hence, to finish the proof that $w$ is well-ordering, it now suffices to show that
\begin{align}
    w(x_1,x_3) + w(x_2,x_4) = g(d_{13}) + g(d_{24}) \leq g(d_{14}) + g(d_{23}) = w(x_1,x_4) + w(x_2,x_3). \label{eq: end estimate}
\end{align}
For this, we set $d \coloneqq d_{23} +  d_{14} - \max\{d_{13},d_{24}\}$, and note that, from inequalities~\eqref{eq:min inequality} and~\eqref{eq:diagonal inequality},
\begin{align*}
    d_{23} \leq d \leq \min\{d_{13},d_{24}\} \leq \max\{d_{13},d_{24}\} \leq d_{14}.
\end{align*}
In particular, there exists $t\in [0,1]$ such that 
\begin{align*}
    d = (1-t) d_{14} + t d_{23}
\quad \mbox{and}\quad \max\{d_{13},d_{24}\} = t d_{14} + (1-t) d_{23},
\end{align*}
and therefore, by using first the non-increasing property and then the convexity of $g$, we have
\begin{align*}
    g(d_{13}) + g(d_{24}) = g(\min\{d_{13},d_{24}\}) + g(\max\{d_{13},d_{24}\}) \leq g(d) + g(\max\{d_{13},d_{24}\}) \leq g(d_{14}) + g(d_{23}),
\end{align*}
which proves~\eqref{eq: end estimate}.

%%%%%%%%%%%%%%%%%%%%%%%%%%%%%%%%%%%%%%%%%%%%%%%%%%

\bigskip
%%%%%%%%%%%%%%%%%%%%%%%%%%%%%%%%%%%%%%%%%%%%%%%%%%
%%%%%%%%%%%%%%%%%%%%%%%%%%%%%%%%%%%%%%%%%%%%%%%%%%
\end{document}